\documentclass[12pt]{article}

\usepackage{enumitem}
\usepackage{mathtools}
\usepackage{xr}

\newcommand{\change}[1]{{{\small \color{black}  #1}}}
\usepackage{url}
\usepackage{amsfonts}
\usepackage{amsmath,amssymb,color}
\usepackage{graphicx}
\usepackage{enumitem}
\usepackage{multirow}
\usepackage{bm}
\usepackage{verbatim}
\usepackage{float}
\usepackage[toc,page]{appendix}
\usepackage[doublespacing]{setspace}
\usepackage{apacite}
\usepackage{natbib}
\usepackage{adjustbox}
\usepackage{rotating}
\usepackage{subcaption}
\usepackage{amsthm}
\usepackage{booktabs,dcolumn,caption}

\usepackage{mathabx}

\usepackage{array}

\usepackage{amsthm}
\usepackage{cleveref}

\numberwithin{equation}{section}
\numberwithin{figure}{section}

\newcolumntype{L}[1]{>{\raggedright\let\newline\\\arraybackslash\hspace{0pt}}p{#1}}
\newcolumntype{C}[1]{>{\centering\let\newline\\\arraybackslash\hspace{0pt}}p{#1}}
\newcolumntype{R}[1]{>{\raggedleft\let\newline\\\arraybackslash\hspace{0pt}}p{#1}}

\newcommand{\bbb}{{\boldsymbol \beta}}

\newcommand{\ppp}{\boldsymbol p}
\newcommand{\qqq}{\boldsymbol q}

\newcommand{\BB}{\mbox{$\mathbf B$}}

\newcommand{\bb}{\mbox{$\mathbf b$}}

\newcommand{\CC}{\mbox{$\mathbf C$}}

\newcommand{\YY}{\mbox{$\mathbf Y$}}
\newcommand{\yy}{\mbox{$\mathbf y$}}

\newcommand{\g}{\mathbf g}
\newcommand{\aaaa}{\mathbf a}

\newcommand{\xx}{\mathbf x}

\newcommand{\argmax}{\operatornamewithlimits{arg\,max}}
\newcommand{\argmin}{\operatornamewithlimits{arg\,min}}

\theoremstyle{plain}
\newtheorem{theorem}{Theorem}[section]
\newtheorem{lemma}[theorem]{Lemma}

\newtheorem{algorithm}{Algorithm}

\newtheorem{remark}{Remark}

\title{\Large Computation for Latent Variable Model Estimation:\\
A Unified Stochastic Proximal Framework}
\date{}

\begin{document}
\maketitle

\doublespacing

\begin{abstract}

Latent variable models have been playing a central role in psychometrics and related fields. In many modern applications, the inference based on latent variable models involves one or several of the following features: (1) the presence of many latent variables, (2) the observed and latent variables being continuous, discrete, or a combination of both, (3) constraints on parameters, and (4) penalties on parameters to impose model parsimony. The estimation often involves maximizing an objective function based on a marginal likelihood/pseudo-likelihood, possibly with constraints and/or penalties on parameters. Solving this optimization problem is highly non-trivial, due to the
complexities brought by the features mentioned above. Although several efficient algorithms have been proposed, there lacks a unified computational framework that takes all these features into account. In this paper, we fill the gap. Specifically, we provide a unified formulation for the optimization problem
and then propose a quasi-Newton stochastic proximal algorithm. Theoretical properties of the proposed algorithms are established. The computational efficiency and robustness are shown by simulation studies under various settings for latent variable model estimation.
\end{abstract}	

\section{Introduction}

Latent variable models have been playing a central role in psychometrics and related fields. Commonly used latent variable models include item response theory models \citep{embretson2000item,reckase2009multidimensional}, latent class models \citep{clogg1995latent,templin2010diagnostic,von2019handbook}, structural equation models \citep{bollen1989structural}, error-in-variable models \citep{carroll2006measurement}, random-effects models \citep{hsiao2014analysis}, and models for missing data \citep{little2019statistical},  where latent variables have different interpretations, such as hypothetical constructs, `true' variables measured with error, unobserved heterogeneity, and  missing data. We refer the readers to \cite{rabe2004generalized} and \cite{bartholomew2011latent} for a comprehensive review of latent variable models.

A latent variable model contains unobserved latent variables and unknown parameters. For example, an item response theory model contains individual-specific latent traits as latent variables and item-specific parameters as model parameters.
Comparing with models without latent variables, such as linear regression and generalized linear regression, the estimation of latent variable models is typically more involved. This estimation problem can be viewed from three perspectives: (1) fixed latent variables and parameters, (2) random latent variables and fixed parameters, and (3) random latent variables and parameters.

The first perspective, i.e.,  fixed latent variables and parameters,
leads to the \textit{joint maximum likelihood} (JML) estimator. This estimator can often be efficiently computed, for example, by an alternating minimization algorithm \citep{birnbaum1968some,chen2018joint,chen2019structured}. Unfortunately, however, the JML estimator is typically statistically inconsistent \citep{neyman1948consistent,andersen1973conditional,haberman1977maximum,ghosh1995inconsistent}, except under some high-dimensional asymptotic regime that is suitable for large-scale applications \citep{chen2018joint,chen2019structured,haberman1977maximum,haberman2004joint}.  Treating both latent variables and parameters as random variables, the third perspective leads to a full Bayesian estimator, for which many \text{Markov chain Monte Carlo} (MCMC) algorithms have been developed \citep[e.g.,][]{beguin2001mcmc,bolt2003estimation, dunson2000bayesian,dunson2003dynamic, edwards2010markov}.

The second perspective, i.e., random latent variables and fixed parameters, essentially follows an \textit{empirical Bayes} (EB) approach \citep{robbins1956empirical,zhang2003compound}. This perspective is the most commonly adopted one \citep{rabe2004generalized}.
Throughout the paper, we refer to estimators derived under this perspective as EB estimators. Both the full-information \textit{marginal maximum likelihood} (MML) estimator \citep{bock1981marginal} and the limited-information \textit{composite maximum likelihood} (CML) estimator \citep{joreskog2001factor,vasdekis2012composite} can be viewed as special cases.
Such estimators involve optimizing an objective function with respective to the fixed parameters, while the objective function
is often intractable due to an integral with respect to the latent variables.
The most commonly used algorithm for this optimization problem is the \textit{expectation-maximization} (EM) algorithm \citep{dempster1977maximum,bock1981marginal}. This algorithm typically requires to iteratively  evaluate numerical integrals with respective to the latent variables, which is often computationally unaffordable when the dimension of the latent space is high.

A high-dimensional latent space is not the only challenge to the computation of EB estimators. Penalties and constraints on parameters
may also involve in the optimization, further complicating the computation. In fact, penalized estimators have become increasingly more popular in latent variable analysis for learning sparse structure, with applications to restricted latent class analysis, exploratory item factor analysis, variable selection in structural equation models, differential item functioning analysis, among others \citep{chen2015statistical,sun2016latent,chen2018robust,lindstrom2020model,tutz2015penalty,jacobucci2016regularized,magis2015detection}. The penalty function is often non-smooth (e.g., Lasso penalty, \citealp{tibshirani1996regression}), for which many standard optimization tools (e.g., gradient descent methods)  are not applicable.  In addition, complex inequality constraints are also commonly encountered in latent variable estimation, for example, in structural equation models \citep{van2010testing} and restricted latent class models \citep[e.g.,][]{de2011generalized,xu2017identifiability}. Such complex constraints further complicate the optimization.


In this paper, we propose a quasi-Newton stochastic proximal algorithm that simultaneously tackles the computational challenges mentioned above. This algorithm can be viewed as an extension of
the stochastic approximation (SA) method \citep{robbins1951stochastic}.
Comparing with SA, the proposed method converges faster and is
more robust, thanks to  the use of Polyak-Ruppert averaging \citep{polyak1992acceleration,ruppert1988efficient}.
The proposed method can also be viewed as a stochastic version of a proximal gradient descent algorithm \citep[Chapter 4,][]{parikh2014proximal}, in which constraints and penalties are handled by a proximal update. As will be illustrated by examples later,  the proximal update is easy to evaluate for many commonly used penalties and constraints, making the proposed algorithm computationally efficient.
Theoretical properties of the proposed method are established, showing that the proposed one is almost optimal in  its convergence speed.


The proposed method is closely related to the stochastic-EM algorithm \citep{celeux1985sem,ip2002single,nielsen2000stochastic,zhang2020improved} and the MCMC stochastic approximation algorithms \citep{cai2010high,cai2010metropolis,gu1998stochastic}, two popular methods for latent variable model estimation. Although these methods perform well in many problems,  they are not as powerful as the proposed one. Specifically,
the MCMC stochastic approximation algorithms cannot handle complex inequality constraints or non-smooth penalties{, because they rely on stochastic gradients which do not always exist when there are complex inequality constraints or non-smooth penalties.}
In addition, as will be discussed later, both
the stochastic-EM algorithm and the MCMC stochastic approximation algorithms are computationally less efficient than the proposed method, even for estimation problems without complex constraints or penalties.

%
%
%
%

The proposed method is also closely related to a perturbed proximal gradient algorithm proposed in \cite{atchade2017perturbed}.
The current development improves upon that of \cite{atchade2017perturbed} from two aspects. First,
the proposed method is a Quasi-Newton method, in which the second-order information (i.e., second derivatives) of the objective function is used in the update. Although this step may only change the asymptotic convergence speed by a constant factor (when the number of iterations grows to infinity), our simulation study suggests that the new method converges much faster than that of \cite{atchade2017perturbed} empirically.
Second, the theoretical analysis of \cite{atchade2017perturbed} only considers a convex optimization setting, while we consider a non-convex setting which is typically the case for latent variable model estimation. Note that the analysis is much more involved when the objective function is non-convex. Therefore, our proof of sequence convergence is different from that of \cite{atchade2017perturbed}. Specifically,
the convergence theory is established
by analyzing the convergence of a set-valued generalization of an ordinary differential equation (ODE).

The rest of the paper is organized as follows. In Section~\ref{sec:problem}, we formulate latent variable model estimation as a general optimization problem which covers many commonly used estimators as special cases. In Section~\ref{sec:algorithm}, a quasi-Newton stochastic proximal algorithm is proposed. Theoretical properties of the proposed algorithm are established in Section~\ref{sec:theory}, suggesting that the proposed algorithm achieves the optimal convergence rate.  The performance of the proposed algorithm is demonstrated and compared with other estimators by simulation studies in Section~\ref{sec:simulation}. We conclude with some discussions in Section~\ref{sec:conclusion}.
An R package has been developed and will be published online upon the acceptance of the current paper.

\section{Estimation of Latent Variable Models}\label{sec:problem}

\subsection{Problem Setup}\label{subsec:problem}

We consider the estimation of a parametric latent variable model. We adopt a general setting, followed by concrete examples
in Sections~\ref{subsec:example1} and \ref{subsec:example2}. Let $\YY$ be a random object representing observable data and let $\yy$ be its realization. For example, in item factor analysis (IFA), $\YY$ represents (categorical) responses to all the items from all the respondents.
A latent variable model specifies the distribution of $\YY$ by introducing a set of latent variables $\boldsymbol\xi \in \Xi$, where $\Xi$ denotes the state space of the latent vector $\boldsymbol\xi$. For example, in item factor analysis, $\boldsymbol\xi$ consists of the latent traits of all the respondents and $\Xi$ is a Euclidian space.
Let $\boldsymbol\beta = (\beta_1, ..., \beta_p)^\top \in \mathcal B$ be a set of parameters in the model, where $\mathcal B$ denotes the parameter space.  The goal is to estimate $\boldsymbol\beta$ given observed data $\yy$.

We consider an EB estimator which takes the form
\begin{equation}\label{eq:lik}
l(\boldsymbol \beta) = \log \left(\int_{\Xi} f(\yy, \boldsymbol\xi\mid \boldsymbol\beta)d \boldsymbol\xi\right),
\end{equation}
where $f(\yy, \boldsymbol\xi\mid \boldsymbol\beta)$ is a complete-data likelihood/pseudo-likelihood function that has an analytic form.
We assume that the objective function $l(\boldsymbol \beta)$ is finite for any $\bbb\in \mathcal B$
and is also smooth in $\boldsymbol \beta$.


The estimator is given by solving the following optimization problem
\begin{equation}\label{eq:obj}
\hat \bbb =   \argmin_{\boldsymbol \beta \in \mathcal{B}}  -l(\boldsymbol \beta) + R(\boldsymbol\beta),
\end{equation}
where $R(\boldsymbol\beta)$ is a penalty function that has an analytic form, such as Lasso, ridge, or elastic net regularization functions. Note that the penalty function often depends on tuning parameters. Throughout this paper, we assume these tuning parameters are fixed and thus do not explicitly indicate them in the objective function~\eqref{eq:obj}. In practice, tuning parameters are often unknown and need to be chosen by cross validation or certain information criterion. We point out that many commonly used estimators take the form of \eqref{eq:obj}, including the MML estimator,  the CML estimator, and regularized estimators based on the MML and CML.  {We also point out that
despite its general applicability to latent variable estimation problems, the proposed method is more useful for complex problems that cannot be easily solved by the classical EM algorithm. For certain problems, such as the estimation of linear factor models and simple latent class models, both the E- and M-step of the EM algorithm have closed-form solutions. In that situation, the classical EM algorithm may be computationally more efficient, though the proposed method can still be used.}

\subsection{High-dimensional Item Factor Analysis} \label{subsec:example1}

Item factor analysis models are commonly used in social and behavioral sciences for analyzing categorical response data. For exposition, we focus on binary response data and point out that the extension to ordinal response data is straightforward. Consider $N$ individuals responding to $J$ binary-scored items. Let $Y_{ij}\in \{0, 1\}$ be a random variable denoting person $i$'s response to item $j$ and let $y_{ij}$ be its realization. Thus, we have $\YY = (Y_{ij})_{N\times J}$ and $\yy = (y_{ij})_{N\times J}$, where $\YY$ and $\yy$ are the generic notations introduced in Section~\ref{subsec:problem} for our data. A comprehensive review of IFA models and their estimation can be found in \cite{chen2020estimation}.


It is assumed that the dependence among an individual's responses is driven by a set of latent factors, denoted by $\boldsymbol \xi = (\xi_{ik})_{N\times K}$, where $\xi_{ik}$ represents person $i$'s $k$th factor. Recall that $\boldsymbol \xi$ is our generic notation for the latent variables in Section~\ref{subsec:problem} and here the state space $\Xi = \mathbb R^{N\times K}$. Throughout this paper, we assume the number of factors $K$ is known. 

An IFA model makes the following assumptions:
\begin{enumerate}
  \item $\boldsymbol \xi_i = (\xi_{i1}, ..., \xi_{iK})^\top$, $i = 1, ..., N$, are independent and identically distributed (i.i.d.) random vectors,  following a multivariate normal distribution $N(\mathbf 0, \boldsymbol\Sigma)$. The diagonal terms of $\boldsymbol\Sigma = (\sigma_{kk'})_{K\times K}$ are set to one for model identification. As $\boldsymbol\Sigma$ is a positive semi-definite matrix, it is common to reparametrize $\boldsymbol\Sigma$ by Cholesky decomposition,
 $$\boldsymbol\Sigma = \BB\BB^\top,$$
      where $\mathbf B = (b_{kk'})_{K\times K}$ is a lower triangular matrix. Let $\mathbf b_k$ be the $k$th row of $\BB$. Then $\Vert \mathbf b_k \Vert = 1$, $k=1, ..., K$, since the diagonal terms of $\boldsymbol\Sigma$ are constrained to value 1.
  \item $Y_{ij}$ given $\boldsymbol \xi_i$ follows a Bernoulli distribution satisfying
                  \begin{equation}\label{eq:logistic}
                    \mathbb{P}(Y_{ij} = 1\mid \boldsymbol \xi_i) = \frac{\exp(d_j + \aaaa_j^\top \boldsymbol\xi_i)}{1+\exp(d_j + \aaaa_j^\top \boldsymbol\xi_i)},
                  \end{equation}
  where $d_j$ and $\aaaa_j = (a_{j1}, ..., a_{jK})^\top$ are item-specific parameters. The parameters $a_{jk}$ are often known as the loading parameters.
  \item $Y_{i1}$, ..., $Y_{iJ}$ are assumed to be conditionally independent given $\boldsymbol\xi_i$, which is known as the local independence assumption.
\end{enumerate}

Note that we consider the most commonly used logistic model in \eqref{eq:logistic}.
It is worth pointing out that the proposed algorithm also applies to the normal ogive (i.e. probit) model which assumes that
$ \mathbb{P}(Y_{ij} = 1\mid \boldsymbol \xi_i) =\Phi(d_j + \aaaa_j^\top \boldsymbol\xi_i)$. Under the current setting and using the reparametrization for $\boldsymbol\Sigma$,  our model parameters are
$\bbb = \{\mathbf B, d_j, \aaaa_j, j =1, ..., J\}$.
The marginal likelihood function takes the form
\begin{equation}\label{eq:IFA}
l(\bbb) = \sum_{i=1}^N \log \left(\int_{\xx \in \mathbb R^{K}} \prod_{j=1}^J \frac{\exp[y_{ij}(d_j + \aaaa_j^\top \boldsymbol x)]}{1+\exp(d_j + \aaaa_j^\top \boldsymbol x)} \phi(\xx\mid \BB) d\xx\right),
\end{equation}
where $\phi(\xx\mid \BB)$ is the density function for multivariate normal distribution $N(\mathbf 0, \BB\BB^\top )$. The $K$-dimensional integrals involved in \eqref{eq:IFA} cause a high computational burden for a relatively large $K$ (e.g., $K \geq 5$).

IFA models are commonly used for both exploratory and confirmatory analyses. In exploratory IFA, an important problem is to learn a sparse loading matrix $(a_{ij})_{J\times K}$ from data, which facilitates the interpretation of the factors.
One approach is by the $L_1$-regularized estimator \citep{sun2016latent} which takes the form
\begin{equation}\label{eq:eifa}
\begin{aligned}
  \hat \bbb =& \argmin_{\bbb \in \mathcal B} - l(\bbb) + R(\bbb),
\end{aligned}
\end{equation}
where the parameter space
$$\mathcal B = \{\bbb \in \mathbb R^{p}:  b_{kk'} = 0, 1\leq k < k' \leq K, \sum_{k' = 1}^K b_{kk'}^2 = 1, k = 1, ..., K\},$$
and the penalty term
\begin{equation}\label{eq:L1}
R(\bbb) = \lambda \sum_{j = 1}^J\sum_{k=1}^K \vert a_{jk}\vert.
\end{equation}
In $R(\bbb)$, $\lambda > 0$ is a  tuning parameter assumed to be fixed throughout this paper.
This regularized estimator resolves the rotational indeterminacy issue in exploratory IFA, as the $L_1$ penalty term is not rotational invariant.
Consequently, under mild regularity conditions, the loading matrix can be consistently estimated only up to a column swapping.
Note that only the $\mathbf B$ matrix has constraints, as reflected by the parameter space $\mathcal B$. Here $b_{kk'} = 0$ is due to that $\mathbf B$ is a lower triangle matrix and $\sum_{k' = 1}^K b_{kk'}^2 = 1$ is due to that the diagonal terms of $\boldsymbol\Sigma = \BB\BB^\top$ are all 1. We remark that it is possible to replace the $L_1$ penalty in $R(\bbb)$ by
other penalty functions for imposing sparsity,
such as the elastic net penalty \citep{zou2005regularization} 
\begin{equation}\label{eq:els}
R(\bbb) =  \lambda_1 \sum_{j = 1}^J\sum_{k=1}^K a_{jk}^2 + \lambda_2 \sum_{j = 1}^J\sum_{k=1}^K \vert a_{jk}\vert,
\end{equation}
where $\lambda_1, \lambda_2 > 0$ are two tuning parameters.


In confirmatory IFA, zero constraints are imposed on loading parameters, based on prior knowledge about the measurement design.
More precisely, these zero constraints can be coded by a binary matrix $\mathbf Q = (q_{jk})_{J\times K}$.  If $q_{jk} = 0$,  then item $j$ does not load on factor $k$ and $a_{jk}$ is set to 0. Otherwise, $a_{jk}$ is freely estimated. These constraints lead to parameter space
$\mathcal B = \{\bbb:  b_{kk'} = 0, 1\leq k < k' \leq K; \sum_{k' = 1}^K b_{kk'}^2 = 1, a_{jk} = 0 \mbox{~for~} q_{jk} = 0, j = 1, ...,J, k = 1, ..., K\}$.
The MML estimator for confirmatory IFA is then given by
\begin{equation}\label{eq:cifa}
\begin{aligned}
  \hat \bbb =& \argmin_{\bbb \in \mathcal B} - l(\bbb).
\end{aligned}
\end{equation}

Besides parameter estimation, another problem of interest in confirmatory IFA is to make statistical inference, for which
it is required
to compute the asymptotic variance of $\hat \bbb$. The estimation of the asymptotic variance often requires to compute the Hessian matrix of $l(\bbb)$ at $\hat \bbb$, which also involves intractable $K$-dimensional integrals.
As we will see in Section~\ref{subsec:algo}, this Hessian matrix, as well as quantities taking a similar form, can be easily obtained as a by-product of the proposed algorithm.


\subsection{Restricted Latent Class Model} \label{subsec:example2}

Our second example is restricted latent class  models which are also widely used in social and behavioral sciences. For example, they are commonly used in education for cognitive diagnosis \citep{von2019handbook}. These models differ from IFA models in that they assume discrete latent variables. Here, we consider a setting for cognitive diagnosis when both data and latent variables are binary.
Consider data taking the same form as that for IFA, denoted by $\YY = (Y_{ij})_{N\times J}$ and $\yy = (y_{ij})_{N\times J}$. In this context, $Y_{ij} = 1$ means that item $j$ is answered correctly and $Y_{ij} = 0$  means an incorrect answer.


The restricted latent class model assumes that each individual is characterized by
a $K$-dimensional latent vector $\boldsymbol \xi_i = (\xi_{i1}, ..., \xi_{iK})^\top$, $i = 1, ..., N$, where $\xi_{ik} \in \{0, 1\}$. Thus, the latent variables are $\boldsymbol\xi = (\xi_{ik})_{N\times K}$, whose state space $\Xi = \{0,1\}^{N\times K}$ contains all $N\times K$ binary matrices.
Each dimension of $\boldsymbol\xi_i$ represents a skill, and $\xi_{ik} = 1$ indicates that person $i$ has mastered the $k$th skill and $\xi_{ik} = 0$ otherwise.

The restricted latent class model can be parameterized as follows.
\begin{enumerate}
  \item The person-specific latent vectors $\boldsymbol\xi_i$, $i = 1, ..., N$, are i.i.d., following a categorical distribution satisfying
  $$\mathbb{P}(\boldsymbol\xi_i = \boldsymbol\alpha) = \frac{\exp(\nu_{\boldsymbol\alpha})}{\sum_{\boldsymbol\alpha' \in \{0, 1\}^K}\exp(\nu_{\boldsymbol\alpha'})},$$
  where  $\boldsymbol\alpha \in \{0, 1\}^K$ represents an attribute profile representing the mastery status on all $K$ attributes, and we set $\nu_{\boldsymbol\alpha'} = 0$ as the baseline, for $\boldsymbol\alpha' = (0, ..., 0)^\top$.
  \item $Y_{ij}$ given $\boldsymbol \xi_i$ follows a Bernoulli distribution, satisfying
  $$\mathbb{P}(Y_{ij} = 1\mid\boldsymbol \xi_i =  \boldsymbol\alpha) = \frac{\exp(\theta_{j, \boldsymbol\alpha})}{1+\exp(\theta_{j, \boldsymbol\alpha})}, ~\boldsymbol\alpha \in \{0, 1\}^K.$$
  \item Local independence is still assumed. That is,  $Y_{i1}$, ..., $Y_{iJ}$ are conditionally independent given $\boldsymbol\xi_i$.
\end{enumerate}
The above model specification leads to a marginal likelihood function
\begin{equation}\label{eq:dcmlik}
l(\bbb) = \sum_{i=1}^N \log \left( \sum_{\boldsymbol\alpha\in \{0,1\}^K} \frac{\exp(\nu_{\boldsymbol\alpha})}{\sum_{\boldsymbol\alpha' \in \{0, 1\}^K}\exp(\nu_{\boldsymbol\alpha'})} \prod_{j=1}^J \frac{\exp(y_{ij}\theta_{j, \boldsymbol\alpha})}{1+\exp(\theta_{j, \boldsymbol\alpha})}  \right),
\end{equation}
where $\bbb = \{\nu_{\boldsymbol\alpha}, \theta_{j, \boldsymbol\alpha}, \boldsymbol\alpha\in \{0, 1\}^K,  j = 1, ..., J\}$.

We consider a confirmatory setting where there exists a design matrix, similar to the $\mathbf Q$-matrix in confirmatory IFA. With slight abuse of notation, we still denote $\mathbf Q = (q_{jk})_{J\times K}$, where $q_{jk} \in \{0, 1\}$. Here, $q_{jk} = 1$ indicates that solving item $j$ requires the $k$th skill and $q_{jk} = 0$ otherwise. As will be explained below, this design matrix leads to equality and inequality constraints in model parameters.

Denote $\mathbf q_{j} = (q_{j1}, ..., q_{jK})^\top$ as the design vector for item $j$. For $\boldsymbol\alpha = (\alpha_1, ..., \alpha_K)^\top$, we write
$$\boldsymbol\alpha \succeq \mathbf q_{j}, \mbox{~if~} \alpha_k \geq  q_{jk} \mbox{~for all~} k \in \{1, ..., K\},$$
and write
$$\boldsymbol\alpha \nsucceq \mathbf q_{j}, \mbox{~if there exists $k$ such that~}  \alpha_k <  q_{jk}.$$
That is, $\boldsymbol\alpha \succeq \mathbf q_{j}$ if profile $\boldsymbol\alpha$ has all the skills needed for solving item $j$ and $\boldsymbol\alpha \nsucceq \mathbf q_{j}$ if not. The design information leads to the following constraints:
\begin{enumerate}
  \item $\mathbb{P}(Y_{ij} = 1\mid\boldsymbol \xi_i =  \boldsymbol\alpha) = \mathbb{P}(Y_{ij} = 1\mid\boldsymbol \xi_i =  \boldsymbol\alpha')$, if both
      $\boldsymbol\alpha, \boldsymbol\alpha' \succeq \mathbf q_{j}$. That is,  individuals who have mastered all the required skills have the same chance of answering the item correctly.
  \item $\mathbb{P}(Y_{ij} = 1\mid\boldsymbol \xi_i =  \boldsymbol\alpha) \geq \mathbb{P}(Y_{ij} = 1\mid\boldsymbol \xi_i =  \boldsymbol\alpha')$
  if $\boldsymbol\alpha\succeq \mathbf q_{j}$ and $\boldsymbol\alpha'\nsucceq \mathbf q_{j}$. That is, students who have mastered all the required skills have a higher chance of answering the item correctly than those who do not.
  \item $\mathbb{P}(Y_{ij} = 1\mid \boldsymbol\xi_i=\boldsymbol \alpha)\geq \mathbb{P}(Y_{ij}=1\mid \boldsymbol \xi_i=\boldsymbol 0)$ for all $\boldsymbol \alpha$. That is, students who have not mastered any skill have the lowest chance of answering correctly.
\end{enumerate}
We refer the readers to \cite{xu2017identifiability} for more discussions on these constraints which are key to the identification of this model. Under these constraints, the MML estimator is given by
\begin{equation}\label{eq:dcm}
\begin{aligned}
  \hat \bbb =& \argmin_{\bbb \in \mathcal B} - l(\bbb), 
%
\end{aligned}
\end{equation}
where
$$
\begin{aligned}
\mathcal B =  \{\bbb:  \max_{\boldsymbol\alpha\succeq \mathbf q_{j}} \theta_{j, \boldsymbol\alpha} = \min_{\boldsymbol\alpha\succeq \mathbf q_{j}} \theta_{j, \boldsymbol\alpha} \geq \theta_{j, \boldsymbol\alpha'} \geq \theta_{j,\boldsymbol 0}, \mbox{~for all~} \boldsymbol\alpha', \nu_{\boldsymbol 0} = 0\}.
  \end{aligned}
  $$

When $K$ is relatively large, the computation for solving \eqref{eq:dcm} becomes challenging, due to both the summation over $2^K$ possible values of $\boldsymbol\alpha$ in $l(\bbb)$, and the large number of inequality constraints.

\section{Stochastic Proximal Algorithm}\label{sec:algorithm}

In this section, we propose a quasi-Newton stochastic proximal algorithm for the computation of \eqref{eq:obj}. The description in this section will focus on the computation aspect, without emphasizing the regularity conditions needed for its convergence.
A rigorous theoretical treatment will be given in Section~\ref{sec:theory}. In what follows, we describe the algorithm in its general form in Section~\ref{subsec:algo}, followed by details for two specific models in Sections~\ref{subsec:exmp1} and \ref{subsec:exmp2}, and finally comparisons with related algorithms in Section~\ref{subsec:connect}.

\subsection{General Algorithm}\label{subsec:algo}

For ease of exposition, we introduce some new notations. We write the penalty function as the sum of two terms,
$R(\bbb) = R_1(\bbb) + R_2(\bbb)$, where $R_1(\bbb)$ is a smooth function and $R_2(\bbb)$ is non-smooth. In the example of regularized estimation for  exploratory IFA,
$R_1(\bbb) = 0$ and $R_2(\bbb) = \lambda \sum_{j = 1}^J\sum_{k=1}^K \vert a_{jk}\vert$,
when $R(\bbb)$ is an $L_1$ penalty as in \eqref{eq:L1}.
When an elastic net penalty is used as in \eqref{eq:els},
$R_1(\bbb) = \lambda_1 \sum_{j = 1}^J\sum_{k=1}^K  a_{jk}^2 $ and $R_2(\bbb) = \lambda_2 \sum_{j = 1}^J\sum_{k=1}^K \vert a_{jk}\vert$.

The optimization problem can be reexpressed as
\begin{equation}\label{eq:obj2}
\min_{\boldsymbol \beta}~  h(\bbb) + g(\bbb),
\end{equation}
where $h(\bbb) = -l(\boldsymbol \beta) + R_1(\boldsymbol\beta)$ and $g: \mathbb  R^{p} \rightarrow \mathbb R \cup \{+\infty\}$ is a generalized function taking the form $g(\bbb) = R_2(\bbb) + I_{\mathcal B}(\bbb)$, where
\begin{equation}\label{eq:h_fun}
I_{\mathcal B}(\bbb) = \left\{\begin{array}{ll}
              0, & \mbox{~if~} \bbb\in \mathcal B,\\
              +\infty, & \mbox{~otherwise.~}
            \end{array}\right.
\end{equation}
Note that since both $l(\bbb)$ and $R_1(\bbb)$ are smooth in $\bbb$,  $h(\bbb)$   is still smooth in $\bbb$. The second term
$g(\bbb)$
is non-smooth in $\bbb$, unless it is degenerate (i.e., $g(\bbb) \equiv 0$). We further write
\begin{equation}\label{eq:H_fun}
H(\boldsymbol\xi, \bbb) = - \log f(\yy, \boldsymbol\xi\mid \bbb) + R_1(\bbb),
\end{equation}
which can be viewed as a complete-data version of $h(\bbb)$ that will be used in the algorithm.


%


The algorithm relies on a scaled proximal operator \citep{lee2014proximal} for the $g$ function, defined as
$$\text{Prox}_{\gamma, g}^{\mathbf D}(\bbb) = \argmin_{\xx\in \mathbb R^{p}} \left\{g(\xx) + \frac{1}{2\gamma} \Vert\xx - \bbb\Vert^2_{\mathbf D} \right\},$$
where $\gamma>0$, $\mathbf D$ is a strictly positive definite matrix, and $\Vert \cdot\Vert_{\mathbf D}$ is a norm defined by $\Vert \xx\Vert_{\mathbf D}^2 = \langle \boldsymbol{x},\boldsymbol{x}\rangle_{\mathbf D}= \xx^\top \mathbf D \xx$.
The choices of $\gamma$, $\mathbf D$, and the intuition behind the proximal operator will be explained in the sequel.

Our general algorithm is described in Algorithm~\ref{alg:alg1}, followed by implementation details. The proposed algorithm is an extension of a
perturbed proximal gradient algorithm \citep{atchade2017perturbed}. The major difference is that the proposed algorithm makes use of
second-order information from the smooth part of the objective function, which can substantially speed up its convergence. See Section~\ref{subsec:connect} for further comparison.

\begin{algorithm}[Stochastic Proximal Algorithm]\label{alg:alg1}~
\begin{itemize}
\item[] \textbf{Input:} Data $\yy$, initial parameters $\bbb^{(0)}$,  a sequence of step size $\gamma_s, s= 1, 2, ...$, pre-specified tuning parameters $c_{2} \geq c_{1} > 0$, and burn-in size $\varpi$.
\item[] \textbf{Update:} At $t$th iteration where $t \geq 1$, we perform the following two steps:
\begin{itemize}
  \item[1.]\textbf{Stochastic step:}  Sample $\boldsymbol\xi$ from the conditional distribution of $\boldsymbol\xi$ given $\yy$,
  $$\psi(\boldsymbol\xi) \propto f(\yy, \boldsymbol\xi\mid \bbb^{(t-1)}),$$
  and obtain $\boldsymbol\xi^{(t)}$. The sampling can be either exact or approximated by MCMC.

  \item[2.]\textbf{Proximal step:}  Update model parameters by
  \begin{equation}\label{eq:prox}
  \bbb^{(t)} = \text{Prox}_{\gamma_t, g}^{\mathbf D^{(t)}}\big(\bbb^{(t-1)} -  \gamma_t(\mathbf D^{(t)})^{-1}\mathbf G^{(t)}\big),
  \end{equation}
  where
     \begin{gather*}
      \mathbf G^{(t)} = \frac{\partial H(\boldsymbol\xi^{(t)}, \bbb)}{\partial \bbb} \bigg\vert_{\bbb = \bbb^{(t-1)}}.
      \end{gather*}
  $\mathbf D^{(t)}$ is a diagonal matrix with diagonal entries $$\delta_i^{(t)} = \frac{t-1}{t} \delta_i^{(t-1)} + \frac{1}{t}T\left(\tilde \delta_i^{(t)}; c_{1}, c_{2}\right),$$
%
  where
  $T(x;c_1,c_2)$ is a truncation function defined as
  \begin{equation}\label{eq:trun}
  T(x;c_1,c_2) = \left\{\begin{array}{ll}
               c_1, & \mbox{~if~} x < c_1, \\
               x, & \mbox{~if~} x \in [c_1,c_2],\\
               c_2, & \mbox{~if~} x > c_2.
             \end{array}\right.
  \end{equation}
  Here $\tilde\delta_i^{(t)} = \tilde\delta_{i,1}^{(t)} + (\tilde\delta_{i,2}^{(t)})^2,$ where
  \begin{gather*}
      \tilde \delta_{i,1}^{(t)} = (1-\gamma_t)\tilde \delta_{i,1}^{(t-1)} + \gamma_t \left(\frac{\partial^2 H(\boldsymbol\xi^{(t)}, \bbb)}{\partial \beta_i^2}\bigg\vert_{\bbb = \bbb^{(t-1)}} - \left(\frac{\partial H(\boldsymbol\xi^{(t)},\bbb)}{\partial\beta_i}\right)^2\bigg\vert_{\bbb=\bbb^{(t-1)}}\right),\\
      \tilde\delta_{i,2}^{(t)} = (1-\gamma_t)\tilde\delta_{i,2}^{(t-1)} + \gamma_t \left(\frac{\partial H(\boldsymbol\xi^{(t)},\bbb)}{\partial\beta_i}\right)\bigg\vert_{\bbb=\bbb^{(t-1)}}.
  \end{gather*}

  %
%
%

\end{itemize}
Iteratively perform these two steps until a stopping criterion is satisfied and let $n$ be the last iteration number.
\item[] \textbf{Output:} $\bar \bbb_n = {\sum_{t=\varpi+1}^n \bbb^{(t)}}/(n-\varpi)$.
\end{itemize}

\end{algorithm}

In what follows, we make a few remarks to provide some intuitions about the algorithm.

\begin{remark}[Connection with stochastic gradient descent] \label{rmk:sgd}
To provide some intuition about the proposed method, we first make a connection between the proposed method and
the stochastic gradient descent (SGD) algorithm.
In fact, when the sampling of $\boldsymbol\xi$ is exact in the stochastic step, then $\mathbf G^{(t)}$ is a stochastic gradient of the smooth part of our objective function, in the sense that
$\mathbb{E}(\mathbf G^{(t)}\mid \yy, \bbb^{(t-1)}) = \nabla h(\bbb)\vert_{\bbb = \bbb^{(t-1)}}.$
If, in addition, there is no constraint or non-smooth penalty, i.e., $g(\bbb) \equiv 0$, then the proximal step degenerates to an SGD update
$\bbb^{(t)} = \bbb^{(t-1)} -  \gamma_t(\mathbf D^{(t)})^{-1}\mathbf G^{(t)}$. In that case, the proposed method becomes a version of SGD.

\end{remark}

\begin{remark}[Proximal step]\label{rmk:prox}
 We provide some intuitions about the proximal step.
We start with two special cases. First, as mentioned in Remark~\ref{rmk:sgd},
if there is no constraint or non-smooth penalty,  then the proximal step is nothing but a stochastic gradient descent step. This is because,
the scaled proximal operator degenerates to an identity map, i.e., $\text{Prox}_{\gamma, g}^{\mathbf D}(\bbb) = \bbb$.
Second, when the $g$ function involves constraints but does not contain a non-smooth penalty, then the proximal step is a projected stochastic gradient descent step. That is, one first performs a stochastic gradient descent update $\tilde \bbb^{(t)} = \bbb^{(t-1)} -  \gamma_t(\mathbf D^{(t)})^{-1}\mathbf G^{(t)}$. Then $\tilde \bbb^{(t)}$ is projected back to the feasible region $\mathcal B$ by the scaled proximal operator:
$$\hat \bbb = \argmin_{\bbb\in \mathcal B} \Vert \bbb - \tilde \bbb^{(t)}\Vert_{\mathbf D},$$
which is a projection under the norm $\Vert\cdot \Vert_{\mathbf D}$. When $\mathbf D$ is an identity matrix as in the vanilla (i.e., non-scaled) proximal operator, then the projection is based on the Euclidian distance.

More generally, when the $g$ function involves non-smooth penalties, then the proximal step can be viewed as minimizing the sum of $g(\bbb)$ and a quadratic approximation of $h(\bbb)$ at $\bbb^{(t-1)}$; see \cite{lee2014proximal} for more explanations. We provide an example to facilitate the understanding. Suppose that
$$g(\bbb) = \lambda \sum_{i=1}^p \vert \beta_i\vert$$
is the Lasso penalty, and $\mathbf D = diag(\delta_1, ..., \delta_p)$ is a diagonal matrix, where $\lambda, \delta_i > 0$, $i = 1, ..., p$.
Then $\text{Prox}_{\gamma, g}^{\mathbf D}(\tilde \bbb^{(t)})$ involves solving $p$ optimization problems separately, each of which takes the form
\begin{equation}\label{eq:soft}
\hat \beta_i = \argmin_{x} \frac{1}{2} (x - \tilde \beta_i^{(t)})^2 + \frac{\lambda\gamma}{\delta_i} |x|.
\end{equation}
It is well known that \eqref{eq:soft} has a closed-form solution given by soft-thresholding \citep[see Chapter 3,][]{friedman2001elements}:
$$\hat \beta_i = \left\{\begin{array}{ll}
                   \tilde \beta_i^{(t)} -  \frac{\lambda\gamma}{\delta_i}, & \mbox{~if~} \tilde \beta_i^{(t)} >  \frac{\lambda\gamma}{\delta_i}, \\
                   \tilde \beta_i^{(t)} + \frac{\lambda\gamma}{\delta_i}, & \mbox{~if~} \tilde \beta_i^{(t)} < -  \frac{\lambda\gamma}{\delta_i}, \\
                   0, & \mbox{~otherwise}.
                 \end{array}\right.$$

\end{remark}

\begin{remark}[Role of $\mathbf D^{(t)}$]
Our proximal step is a quasi-Newton proximal update proposed in \cite{lee2014proximal} under a non-stochastic optimization setting. As shown in \cite{lee2014proximal}, quasi-Newton proximal methods converge faster than first-order proximal methods under the non-stochastic setting. Here, the diagonal matrix $\mathbf D^{(t)}$ is used to approximate the Hessian matrix of
$h(\bbb)$ at $\bbb^{(t)}$. When $\bbb^{(t)}$ converges to $\hat \bbb$, then
$\delta_i^{(t)}$,
the $i$th diagonal term of $\mathbf D^{(t)}$,
converges to $T\left(\frac{\partial^2 h}{\partial \beta_i^2}\vert_{\bbb = \hat \bbb}; c_{1}, c_{2}\right)$ where $T$ is the truncation function defined in \eqref{eq:trun}; see Remark~\ref{rmk:biprod} for more explanations.

In the proposed update, we choose $\mathbf D^{(t)}$ to be a diagonal matrix for computational convenience. Specifically, as discussed in Remark~\ref{rmk:prox}, the proximal step is in a closed form when $\mathbf D^{(t)}$ is a diagonal matrix. In addition, the proximal step requires to calculate the inverse of $\mathbf D^{(t)}$, whose complexity is much lower when $\mathbf D^{(t)}$ is diagonal.

We point out that using a diagonal matrix to approximate the Hessian matrix is a popular and effective trick in numerical optimization \citep[e.g., Chapter 5,][]{bertsekas1992data,becker1988improving}, especially for large-scale optimization problems.
In principle, it is possible to allow $\mathbf D^{(t)}$ to be non-diagonal. In fact, it is not difficult to generalize the BFGS updating formula for
 $\mathbf D^{(t)}$ given in \cite{lee2014proximal} to a stochastic version.

Our choice of $\mathbf D^{(t)}$ guarantees its eigenvalues
to be constrained in the interval $[c_{1}, c_{2}]$.
It rules out the singular situation when $\mathbf D^{(t)}$ is not strictly positive definite.
In the implementation, we set $c_{1}$'s to be a sufficiently small constant and set $c_{2}$'s to be a sufficiently large constant.
According to simulation, the algorithm tends to be insensitive to these choices.


\end{remark}


We further provide some remarks regarding the implementation details.

\begin{remark}[Choices of step size] \label{rmk:step}
As will be shown in Section~\ref{sec:theory},
the convergence of the proposed method requires the step size to satisfy
$\sum_{t=1}^{\infty} \gamma_t = \infty$ and $\sum_{t=1}^{\infty} \gamma_t^2 < \infty$.
This requirement is also needed in the Robbins-Monro algorithm. Here, we choose the step size $\gamma_t = \mu t^{-\frac{1}{2} - \varepsilon}$ so that the above requirement is satisfied, where
$\mu$ is a positive constant and $\varepsilon$ is a small positive constant. As will be shown in Section~\ref{sec:theory}, with sufficiently small $\varepsilon$, $\bar \bbb_n$ is almost optimal in terms of its convergence speed. We point out that $\varepsilon$ is needed to prove the convergence of $\bar \bbb_n$, under our non-convex setting. It is not needed, if the objective function~\eqref{eq:obj} is convex; see \cite{atchade2017perturbed}. The requirement of $\varepsilon$ may be an artifact due to our proof strategy. Simulation results show that the algorithm converges well even if we set $\varepsilon = 0$.
For the numerical analysis in this paper, we set $\varepsilon = 10^{-2}$.

We point out that our choice of step size
is very different from the step size in the Robbins-Monro algorithm, for which asymptotic results \citep{fabian1968asymptotic} suggest that the optimal choice of step size satisfies
$\gamma_t = O(1/t)$.



\end{remark}

\begin{remark}[Starting point] As the objective function~\eqref{eq:obj} is typically non-convex for most latent variable models, the choice of the starting point $\bbb^{(0)}$ matters. The algorithm is more likely to converge to the global optimum given
a good starting point. One strategy is to run the proposed algorithm with multiple random starting points and then choose the best-fitting solution. Alternatively, one may find a good starting point using less accurate but computationally faster estimators, such as the
constrained joint maximum likelihood estimator \citep{chen2018joint,chen2019structured} or spectral methods \citep{zhang2019note}.
 {Moreover, to further avoid convergence  to local optima, one may also use multiple random starting points and choose the one with the smallest objective function value.}

\end{remark}

\begin{remark}[Sampling in stochastic step]
As mentioned in Remark~\ref{rmk:sgd}, when the latent variables $\boldsymbol\xi$ can be sampled exactly in the stochastic step, then
$\mathbf G^{(t)}$ is a stochastic gradient of $h(\bbb)$.
Unfortunately, exact sampling is only possible under some situations such as restricted latent class analysis. In most cases,
we only have
approximate samples from an MCMC algorithm. For example,
as discussed below, the latent variables in IFA can be sampled by a block-wise Gibbs sampler. With approximate samples, $\mathbf G^{(t)}$ is only approximately unbiased.
As we show in Section~\ref{sec:theory}, such $\mathbf G^{(t)}$ may still yield convergence of $\bar \bbb_n$.


\end{remark}

\begin{remark}[Stopping criterion]\label{stopping_cre} In the implementation of Algorithm~\ref{alg:alg1}, we stop the iterative update by monitoring
a window of successive differences in $\bbb^{(t)}$. More precisely, we stop the iteration if all differences in the
window are less than a given threshold. Unless otherwise stated, the numerical analysis in this paper uses a window size 3.  {The same stopping criterion is also adopted by the Metroplis-Hasting Robins-Monro algorithm proposed by \cite{cai2010high}.}



\end{remark}

Finally, as we explain in Remark~\ref{rmk:biprod}, certain quantities, including the Hessian matrix of $l(\bbb)$, can be obtained as a by-product of the proposed algorithm.

\begin{remark}[By-product]\label{rmk:biprod}
It is often of interest to compute quantities of the form
\begin{equation}\label{eq:info}
\hat M = \mathbb{E}\left[m(\yy, \boldsymbol\xi\mid \bbb)\mid \yy, \bbb\right]\big\vert_{\bbb = \hat \bbb},
\end{equation}
where $m(\yy, \boldsymbol\xi\mid \bbb)$ is a given function with an analytic form and the conditional expectation $\mathbb{E}\left[\cdot\mid \yy, \bbb\right]$ is with respect to the conditional distribution of $\boldsymbol\xi$ given $\yy$.
The quantity \eqref{eq:info}
is intractable due to the high-dimensional integral with respect to $\boldsymbol\xi$.
One such example is the Hessian matrix
of $l(\boldsymbol \beta)$ at $\hat \bbb$ as discussed in Section~\ref{subsec:example1} that is a key quantity for the statistical inference of
$\hat \bbb$. In fact, by Louis' formula \citep{louis1982finding},
 $$
\begin{aligned}
\frac{\partial^{2} l(\boldsymbol{\beta})}{\partial \boldsymbol{\beta} \partial \boldsymbol{\beta}^{\top}}=& \mathbb{E}\left[ \left.\frac{\partial^{2} (\log f(\yy, \boldsymbol\xi\mid \boldsymbol\beta))}{\partial \bbb \partial \bbb^{\top}} + \frac{\partial (\log f(\yy, \boldsymbol\xi\mid \boldsymbol\beta))}{\partial \bbb}\left[\frac{\partial (\log f(\yy, \boldsymbol\xi\mid \boldsymbol\beta))}{\partial \bbb}\right]^{\top} \right\vert \yy, \bbb\right]\\
& - \mathbb{E}\left[\left. \frac{\partial (\log f(\yy, \boldsymbol\xi\mid \boldsymbol\beta))}{\partial \bbb} \right\vert \yy, \bbb \right] \left(\mathbb{E}\left[\left. \frac{\partial (\log f(\yy, \boldsymbol\xi\mid \boldsymbol\beta)))}{\partial \bbb} \right\vert \yy, \bbb \right]\right)^\top.
\end{aligned}
$$


The computation of \eqref{eq:info} is a straightforward by-product of the proposed algorithm. To approximate $\hat M$, we only need to add the following update in each iteration
\begin{equation}
M^{(t)} = M^{(t-1)} + \gamma_t\big(m(\yy,\boldsymbol\xi^{(t)}\mid \bbb^{(t)})-M^{(t-1)}\big),
\end{equation}
for $t \geq 2$, where $M^{(1)} = m(\yy,\boldsymbol\xi^{(1)}\mid \bbb^{(1)})$. We approximate $\hat M$ by the Polyak-Ruppert averaging
$\bar M_n = (\sum_{t=\varpi + 1}^n M^{(t)})/(n-\varpi)$. When the sequence $\bbb^{(t)}$ converges to $\hat \bbb$ (see Theorem~\ref{thm:thm1} for the convergence analysis), under mild conditions, Theorem~3.17 of \cite{benveniste1990adaptive} suggests the convergence of $M^{(n)}$ to $\hat M$ with probability 1, which further implies the convergence of $\bar M_n$ to $\hat M$. 
Note that we use the averaged estimator $\bar M_n$ as it tends to converge faster than the pre-average  sequence $M^{(n)}$.
We point out that the updating rule for the diagonal matrix $\mathbf D^{(t)}$ in Algorithm~\ref{alg:alg1} makes use of such an averaged estimator.

\begin{remark}[Burn-in size]

Like MCMC algorithms, the proposed method also has a burn-in period, where parameter updates from that period are not used in the Polyak-Ruppert averaging. The choice of the burn-in size will not affect the asymptotic property of the method, but does affect the empirical performance. This is because, the parameter updates may be far away from the solution due to the effect of the starting point. Including them in the Polyak-Ruppert averaging may introduce a high bias.
In our numerical analysis, the burn-in size $\varpi$ is fixed to be sufficiently large in each of our examples.
Adaptive choice of the burn-in size is possible; see \cite{zhang2020improved}.

\end{remark}

\end{remark}

\subsection{Example I: Item Factor Analysis}\label{subsec:exmp1}

We now explain the details of using the proposed method to solve \eqref{eq:eifa} for exploratory IFA. The computation is similar when replacing the $L_1$ regularization by the elastic net regularization.
For confirmatory IFA,
the stochastic step is the same as that of exploratory IFA and the proximal update step is straightforward as no penalty is involved.
Therefore, the details for the computation of confirmatory IFA are omitted here.

We first consider the stochastic step for solving \eqref{eq:eifa}. Note that $\boldsymbol\xi_1$, ..., $\boldsymbol\xi_N$ are conditionally independent given data, and thus can be sampled separately.
For each $\boldsymbol\xi_i$, we sample its entries by Gibbs sampling. More precisely, each entry is sampled by
adaptive rejection sampling \citep{gilksAdaptiveRejectionSampling1992,zhang2020improved},
as the conditional distribution of $\xi_{ik}$ given data and the other entries of $\boldsymbol\xi_i$ is log-concave. We refer the readers to \cite{zhang2020improved} for more explanations of this sampling procedure.
If a normal ogive IFA is considered instead of the logistic model above, then we can sample $\boldsymbol\xi^{(t)}_i$ by a similar Gibbs method with a data augmentation trick; see \cite{chen2020estimation} for a review.

We now discuss the computation for the proximal step. Recall that
$\bbb = \{\mathbf B, d_j, \aaaa_j, j =1, ..., J\}$. We denote $$\tilde \bbb^{(t)} = \bbb^{(t-1)} -  \gamma_t(\mathbf D^{(t)})^{-1}\mathbf G^{(t)}$$ as the input of the scaled proximal operator. The parameter update is given by
$$\bbb^{(t)} = \text{Prox}_{\gamma, g}^{\mathbf D}(\tilde \bbb^{(t)}) = \argmin_{\bbb} \left\{g(\bbb) + \frac{1}{2\gamma_t} \sum_{i=1}^p \delta_i^{(t)}(\beta_i - \tilde\beta_i )^2\right\},$$
where the parameter space $$\mathcal B = \{\bbb \in \mathbb R^{p}:  b_{kk'} = 0, 1\leq k < k' \leq K, \sum_{k' = 1}^K b_{kk'}^2 = 1, k = 1, ..., K\},$$
and
$g(\bbb) = \lambda \sum_{j = 1}^J\sum_{k=1}^K \vert a_{jk}\vert + I_{\mathcal{B}}(\bbb)$
only involves loading parameters $a_{jk}$ and parameters $\mathbf B$ for the covariance matrix.

We first look at the update for $d_j$s. As the $g$ function does not involve $d_j$, its update is simply
$d_j^{(t)} =  \tilde d_j^{(t)}$, where $\tilde d_j^{(t)}$ is the corresponding component in $\tilde \bbb^{(t)}$. We then look at the update for the loading parameters $a_{jk}$. Suppose that $a_{jk}$ corresponds to the $i_{a_{jk}}$th component of $\bbb$. Then the update is given by solving the optimization
$$a_{jk}^{(t)} = \argmin_{a_{jk}}~ \lambda |a_{jk}| + \frac{1}{2\gamma_t} \delta_{i_{a_{jk}}}^{(t)} (a_{jk} - \tilde a_{jk}^{(t)})^2.$$
As discussed in Remark~\ref{rmk:prox}, this optimization has a closed-form solution via soft-thresholding.
We finally look at the update for $\mathbf B$.
Suppose that $b_{kl}$ corresponds to the $i_{b_{kl}}$th component of $\bbb$.
Then the update of  $\bb_k$,
the $k$th row of $\mathbf B$, is given by solving the following optimization problem:
$$\bb_k^{(t)} = \argmin_{\mathbf{b}_k: \Vert\mathbf{b}_k\Vert=1, b_{kk'} = 0, k' > k} \left\{ \sum_{l=1}^K \delta_{i_{b_{kl}}}^{(t)} (b_{kl} - \tilde b_{kl}^{(t)})^2 \right\},$$
which can be easily solved by the method of Lagrangian multiplier.

\subsection{Example II: Restricted LCA}\label{subsec:exmp2}

We now provide a brief discussion on the computation for the restricted LCA model. First, the stochastic step is straightforward, as the posterior distribution for each $\boldsymbol\xi_i$ is still a categorical distribution which can be sampled exactly. Second, the proximal step requires to solve a quadratic programming problem. Again, we denote
$$\tilde \bbb^{(t)} = \bbb^{(t-1)} -  \gamma_t(\mathbf D^{(t)})^{-1}\mathbf G^{(t)}.$$
The proximal step requires to solve the following quadratic programming problem
\begin{equation}\label{eq:dcm2}
\begin{aligned}
& \min_{\bbb} (\bbb - \tilde \bbb^{(t)})^\top \mathbf D^{(t)} (\bbb - \tilde \bbb^{(t)}), \\
  s.t.~& \max_{\boldsymbol\alpha\succeq \mathbf q_{j}} \theta_{j, \boldsymbol\alpha} = \min_{\boldsymbol\alpha\succeq \mathbf q_{j}} \theta_{j, \boldsymbol\alpha} \geq \theta_{j, \boldsymbol\alpha'}\geq \theta_{j,\boldsymbol 0}, \mbox{~for all~} \boldsymbol\alpha'\nsucceq \boldsymbol q_j,\\
    &  \mbox{~and~} \nu_{\mathbf 0} = 0.
\end{aligned}
\end{equation}
Quadratic programming is the most studied nonlinear convex optimization problem \citep[Chapter 4,][]{boyd2004convex} and many efficient solvers exist. In our simulation study in Section~\ref{subsec:simlca}, we use the dual method of \cite{goldfarb1983numerically} implemented in the R package \textsf{quadprog} \citep{quadprog}.


\subsection{Comparison with Related Algorithms} \label{subsec:connect}

We compare Algorithm~\ref{alg:alg1} with several related algorithms in more details.

\paragraph{Robbins-Monro SA and variants.} The proposed method is closely related to the stochastic approximation approach first proposed in \cite{robbins1951stochastic}, and its variants given in \cite{gu1998stochastic} and \cite{cai2010high} that are specially designed for latent variable model estimation. Note that
the Robbins-Monro method is the first SGD method with convergence guarantee.
Both the methods of \cite{gu1998stochastic} and \cite{cai2010high} approximate the original Robbins-Monro method by using MCMC sampling to generate an approximate stochastic gradient in each iteration, when an unbiased stochastic gradient is difficult to obtain.
All these methods do not handle complex constraints or non-smooth objective functions.

When there is no constraint or penalty on parameters (i.e., $g(\bbb) \equiv 0$), the proximal operator degenerates to an identity map.
In this case, the proposed method is essentially the same as \cite{gu1998stochastic} and \cite{cai2010high}, except for the sampling method in the stochastic step, the way the Hessian matrix is approximated, the specific choices of step size, and the averaging in the last step of the proposed method. Among these differences, the step size and the trajectory averaging are key to the advantage of the proposed method.

As pointed out in Remark~\ref{rmk:step}, the Robbins-Monro procedure  has the same general requirement on the step size as the proposed method. Specially, the Robbins-Monro procedure, as well as its MCMC variants \citep{gu1998stochastic,cai2010high}, typically let the step size $\gamma_t$ decay in the order $1/t$ as suggested by asymptotic theory \citep{fabian1968asymptotic}. However, this step is often too short at the early stage of the algorithm, resulting in poor performance in practice \citep[Section 4.5.3.,][]{spall2003introduction}.
On the other hand, the proposed method adopts a longer step size. By further adopting Polyak-Ruppert averaging \citep{ruppert1988efficient,polyak1992acceleration}, we show in Section~\ref{sec:theory} that the proposed method almost achieve the optimal convergence speed.

%
%

\paragraph{Perturbed proximal gradient algorithm.} Proximal gradient descent algorithm \citep{parikh2014proximal}
is a non-stochastic algorithm for solving nonsmooth and/or constrained optimization algorithms.
For example, the widely used gradient projection algorithm for oblique rotation in factor analysis \citep{jennrich2002simple} is a special case. The vanilla proximal gradient descent algorithm does not use the second-order information of the objective function and thus sometimes converges slowly. To improve convergence speed, proximal Newton-type methods have been proposed in
\cite{lee2014proximal} that utilize the second-order information of the smooth part of the objective function.

The perturbed proximal gradient algorithm \citep{atchade2017perturbed} solves a similar optimization problem as in \eqref{eq:obj}
by combining the methods of stochastic approximation, proximal gradient decent, and Polyak-Ruppert averaging.
The proposed method extends \cite{atchade2017perturbed} by adopting a Newton-type proximal update suggested in \cite{lee2014proximal}.
The method of
\cite{atchade2017perturbed} can be viewed as a special case of the proposed one with $c_{1} = c_{2}$.
As shown by simulation study in the sequel, thanks to the second-order information, the proposed method converges much faster than that of
\cite{atchade2017perturbed}. We also point out that the theoretical analysis of \cite{atchade2017perturbed} focuses on convex optimization, while in Section~\ref{sec:theory} we consider a more general setting of non-convex optimization that includes a wide range of latent variable model estimation problems as special cases.

\paragraph{Stochastic EM algorithm.} The proposed method is also closely related to the stochastic-EM algorithm \citep{celeux1985sem,ip2002single,nielsen2000stochastic,zhang2020improved}.
The stochastic-EM algorithm is a similar iterative algorithm, consisting of a stochastic step and a maximization step in each iteration, where the stochastic step is the same as that in the proposed algorithm.
The maximization step plays a similar role as the proximal step in the proposed algorithm. More precisely, when there is no constraint or penalty, the maximization step of the stochastic-EM algorithm obtains parameter update $\bbb^{(t)}$  by minimizing the negative complete data log-likelihood function $-\log f(\yy, \boldsymbol\xi^{(t)}\mid \bbb)$, instead of
a stochastic gradient update.
It is also recommended to perform a trajectory averaging in the stochastic-EM algorithm \citep{nielsen2000stochastic,zhang2020improved}, like the last step of the proposed algorithm. As pointed out in \cite{zhang2020improved}, the stochastic EM algorithm
can potentially handle constraints and non-smooth penalties on parameters by incorporating them into the maximization step.

The stochastic-EM algorithm  is typically not as fast as the proposed method, which is revealed by simulation studies below. This is because, it requires to solve an  optimization problem completely in each iteration, which is time consuming, especially when constraints and non-smooth penalties are involved.
On the other hand,
the proximal step of the proposed algorithm can often be efficiently performed.


\section{Theoretical Properties} \label{sec:theory}

In what follows, we establish the asymptotic properties of the proposed algorithm{, under suitable technical conditions. For readers who are not interested in the asymptotic theory, this section can be skipped without affecting the reading of the rest of the paper}. Note that in this section, we view data as fixed and the randomness comes from sampling of the latent variable. The following expectation is taken with respect to latent variable $\boldsymbol\xi$ given data $\boldsymbol y$ and parameters $\bbb,$ denoted by $\mathbb{E}(\cdot\mid \bbb)=\int \cdot \pi_{\bbb}(\boldsymbol\xi)d\boldsymbol\xi$, where $\pi_\bbb$ is the posterior distribution for $\boldsymbol\xi$ given $\boldsymbol y$ and $\boldsymbol \bbb.$ Let $\Vert\cdot\Vert$ denote the vector $l_2$-norm. Following the typical convergence analysis of non-convex optimization \citep[e.g., Chapter 3,][]{floudas1995nonlinear}, we will first discuss the convergence of the sequence $\bbb^{(t)}$ to a stationary point of the objective function $h(\bbb) + g(\bbb)$ in Theorem~\ref{thm:thm1}, which follows the theoretical development in \cite{duchiStochasticMethodsComposite2018}.
Then with some additional assumptions on the local geometry of the objective function at the stationary point being converged to, we will show the convergence rate of the Polyak-Ruppert averaged sequence $\bar \bbb_n$ in Theorem~\ref{thm:thm2} which extends the results of \cite{atchade2017perturbed} to the setting of non-convex optimization.


For a function $f:\mathbb{R}^d\mapsto \mathbb{R}\cup +\infty,$ denote the Fr\'{e}chet subdifferential \citep[Chapter 8.B][]{rockafellar1998variational} of $f$ at the point $\boldsymbol x$ by $\partial f(\boldsymbol x),$
\[
  \partial f(\boldsymbol x)=\left\{\boldsymbol z \in \mathbb{R}^{p}: f(\boldsymbol y) \geq f(\boldsymbol x)+ \boldsymbol z^\top(\boldsymbol y-\boldsymbol x)+o(\|\boldsymbol y-\boldsymbol x\|) \text { as } \boldsymbol y \rightarrow \boldsymbol x\right\}.
\]
Define the set of stationary points of the objective function as
$$\mathcal B^* = \{\bbb \in \mathcal B: \exists\ \mathbf x \in \partial h(\bbb) + \partial g(\bbb) \mbox{~with~} \xx^\top (\yy -\bbb) \geq 0, \mbox{~for all~} \yy \in \mathcal B\}.$$
Note that the global minimum $\hat \bbb$ is a stationary point, i.e., $\hat \bbb \in \mathcal B^*$. In addition, when the objective function is smooth, i.e. $g(\bbb) \equiv 0$, then
$\mathcal B^* = \{\bbb \in \mathcal B: \nabla h(\bbb) = 0 \},$
which is the standard definition of stationary points set for a smooth function.

The following assumptions are assumed for our objective function.

\begin{itemize}
    \item[H1.] $\mathcal{B}$ is compact and contains finite stationary points. For stationary points $\bbb,\bbb^\prime\in\mathcal{B}^*,$ $h(\bbb)+ g(\bbb)= h(\bbb^\prime)+g(\bbb^\prime)$ if and only if $\bbb = \bbb^\prime.$
    \item[H2.] $H(\boldsymbol\xi,\bbb)$ is a differentiable function with resepect to $\bbb$ for given $\boldsymbol\xi$ and let $\mathbf G_{\bbb}(\boldsymbol \xi)= \partial H(\boldsymbol\xi,\bbb)/\partial \bbb.$ Define function $M_\epsilon$: $\Theta\times\Xi\rightarrow \mathbb{R}_+$ as
    \[
    M_\epsilon(\bbb;\boldsymbol\xi) = \sup_{\bbb'\in\mathcal{B},\Vert\bbb'-\bbb\Vert<\epsilon}\Vert\mathbf G_{\bbb'}(\boldsymbol \xi)\Vert.
    \]
    There exists $\epsilon_0>0$ such that for all $0<\epsilon<\epsilon_0,$
    \[
    \mathbb{E}[M_\epsilon(\bbb;\boldsymbol\xi)^2\mid \bbb]<\infty \text{ for all }\bbb\in\mathcal{B}.
    \]
    \item[H3.] There exists $\epsilon_0>0$ such that for all $\bbb^\prime\in\mathcal{B},$ there exists $\lambda(\boldsymbol\xi,\bbb^\prime)\geq 0$ such that
    \[
    \bbb \mapsto H(\boldsymbol\xi,\bbb) + \frac{\lambda(\boldsymbol\xi,\bbb^\prime)}{2}\Vert\bbb-\bbb_0\Vert^2
    \]
    is convex on the set $\{\bbb: \Vert\bbb-\bbb^\prime\Vert\leq\epsilon_0\}$ for any $\bbb_0,$ and $\mathbb{E}[\lambda(\boldsymbol\xi,\bbb^\prime)\mid \boldsymbol\bbb]<\infty.$
    \item[H4.] The stochastic gradient $\mathbf G_{\bbb^{(t-1)}}(\boldsymbol\xi^{(t)})$
    is a Monte Carlo approximation of $\nabla h(\bbb^{(t-1)}).$ That is,
    if computationally feasible, we take $\boldsymbol\xi^{(t)}$ as an exact sample from $\pi_{\bbb^{(t-1)}}$, where, as defined earlier, $\pi_{\bbb^{(t-1)}}$ is the posterior distribution of $\boldsymbol\xi$ given $\boldsymbol y$ and $\bbb^{(t-1)}$. If not, we sample $\boldsymbol\xi^{(t)}$ from
    a Markov kernel $P_{\bbb^{(t-1)}}$ with invariant distribution $\pi_{\bbb^{(t-1)}}$.
    \item[H5.] Define
\begin{equation}
    \begin{aligned}
        &\bbb_{\gamma}^+(\boldsymbol\xi) = \underset{\boldsymbol x \in \mathcal{B}}{\operatorname{argmin}}\left\{[\boldsymbol G_{\bbb}(\boldsymbol\xi)]^\top(\boldsymbol x-\bbb)+\g(\boldsymbol x)+\frac{1}{2 \gamma}\left\|\boldsymbol x-\bbb\right\|^{2}_{\mathbf D}\right\},\\
        &\boldsymbol{U}_{\gamma}(\boldsymbol\xi;\bbb) = \frac{1}{\gamma}(\bbb - \bbb_{\gamma}^+(\boldsymbol\xi)),
    \end{aligned}
\end{equation}
    where step size satisfy $\sum_{t=1}^\infty\gamma_t = \infty,$ $\sum_{t=1}^\infty\gamma_t^2<\infty.$ Then with probability 1, $$\lim_{n\rightarrow \infty}\sum_{t=1}^n\gamma_t\left(\boldsymbol{U}_{\gamma_t}(\boldsymbol\xi^{(t)};\bbb^{(t-1)}) - \mathbb{E}[\boldsymbol{U}_{\gamma_t}(\boldsymbol\xi^{(t)};\bbb^{(t-1)})\mid \bbb^{(t-1)}]\right)$$ exists and is finite.
\end{itemize}

We remark that conditions H1 through H5 are quite mild. Condition H1 imposes mild requirements on the compactness of the parameter space and the properties of the stationary points of the objective function.  Specifically, the compactness of the parameter space is often assumed when analyzing stochastic optimization problems without assuming convexity; see e.g., \cite{gu1998stochastic}, \cite{nielsen2000stochastic}, \cite{cai2010metropolis}, and \cite{duchiStochasticMethodsComposite2018}.   {It also requires that the objective function has different values at different stationary points.}
Conditions H2 and H3 require the complete-data log-likelihood function $H(\boldsymbol\xi,\cdot)$ is locally Lipschitzian and weakly convex, respectively.   {These conditions hold when the complete-data log-likelihood function $H(\boldsymbol\xi,\cdot)$ is Lipschitzian and convex on the entire parameter space.
Requiring locally Lipschitzian and weakly convex enables our theory to be applicable to
 a wider range of problems.}  Similar conditions are imposed in \cite{duchiStochasticMethodsComposite2018}. For the examples that we consider in Sections \ref{subsec:example1} and \ref{subsec:example2}, these two conditions are satisfied because $H(\boldsymbol\xi,\bbb)$ is smooth and convex in $\bbb$.  Condition H4 is automatically satisfied according to the way the latent variables are sampled in Algorithm~\ref{alg:alg1}. Finally,
H5 is a key condition for the convergence of the sequence $\bbb^{(t)}.$ When exact samples from the posterior distribution are used,
Lemma~\ref{lem:lem1} below guarantees that H5 is satisfied. With approximate samples from an MCMC algorithm, H5 may still hold when the bias from the MCMC samples is small.


%
\begin{lemma}\label{lem:lem1}
    Define the filtration of $\sigma$-algebra $\mathcal{F}_{t-1} = \sigma\left(\bbb^{(0)}, \boldsymbol\xi^{(k)}, 0 \leq k \leq t-1\right).$ $\boldsymbol\xi$ is a sample from $\pi_\bbb.$ Let $$\boldsymbol\epsilon_{\gamma}(\boldsymbol\xi;\bbb) = \boldsymbol{U}_{\gamma}(\boldsymbol\xi;\bbb) - \mathbb{E}[\boldsymbol{U}_{\gamma}(\boldsymbol\xi;\bbb)\mid\bbb],$$ then $\gamma_t\boldsymbol{\epsilon}_{\gamma_t}(\boldsymbol\xi^{(t)},\bbb^{(t-1)})$ is a square-integrable martingale difference sequence adpated to $\mathcal{F}_{t-1},$ and with probability 1, $\lim_n \sum_{t=1}^n \gamma_t \boldsymbol{\epsilon}_{\gamma_t}(\boldsymbol\xi^{(t)},\bbb^{(t-1)})$ exists and is finite.
\end{lemma}


\begin{theorem}\label{thm:thm1}
  Apply Algorithm~\ref{alg:alg1} to optimization problem~(\ref{eq:obj2}) with step size $\gamma_t=t^{-\frac{1}{2}-\epsilon},\epsilon\in (0,\frac{1}{2}]$, for which conditions H1-H5 hold. Then with probability 1, the sequence $\bbb^{(n)}$ converges to a stationary point in $\mathcal B^*$.
\end{theorem}

We remark that the convergence of the proposed method is similar to that of the EM algorithm. In fact, for marginal maximum likelihood estimation that is non-convex, the EM algorithm also only guarantees
the convergence to
a stationary point \citep{wu1983convergence}. Moreover, when the objective function has a single stationary point (e.g., when the objective function is strictly convex), then Theorem~\ref{thm:thm1} guarantees global convergence.

The convergence of $\bbb^{(n)}$ guarantees the convergence of the Polyak-Ruppert averaging sequence $\bar \bbb_n$. However, Theorem~\ref{thm:thm1} does not provide information on the convergence speed. In what follows, we establish the convergence speed of $\bar \bbb_n$. Without loss of generality, by Theorem~\ref{thm:thm1}, we assume that $\bbb^{(n)}$ converges to $\bbb_* \in \mathcal B^*.$

\begin{itemize}
    \item[H6.] There exists $\delta >0$, such that $h(\bbb)$ is strongly convex in $\mathcal{B}_1 = \{\bbb \in \mathcal B: \Vert \bbb - \bbb_*\Vert \leq \delta\}$ and $\nabla h(\bbb)$ is Lipschitz in $\mathcal{B}_1$ with Lipschitz constant $L$.
    \item[H7.] For $\bbb,\bbb' \in \mathcal{B}_1,$ any $\gamma>0,$ and diagonal matrix $\mathbf D$ with diagonal entries $\delta_i\in[c_1,c_2],$ the following conditions hold.
    \begin{itemize}
        \item[(i)]\label{eq:g_bound1}
        $g\left(\operatorname{Prox}^{\mathbf D}_{\gamma, g}(\bbb)\right)-g\left(\bbb^{\prime}\right) \leq-\frac{1}{\gamma}\left\langle\operatorname{Prox}^{\mathbf D}_{\gamma, g}(\bbb)-\bbb^{\prime}, \operatorname{Prox}^{\mathbf D}_{\gamma, g}(\bbb)-\bbb\right\rangle_{\mathbf D}$.
        \item[(ii)]\label{eq:g_bound2}
        $\left\|\operatorname{Prox}^{\mathbf D}_{\gamma, g}(\bbb)-\operatorname{Prox}^{\mathbf D}_{\gamma,{} g}\left(\bbb^{\prime}\right)\right\|_{\mathbf D}^{2}+\left\|\left(\operatorname{Prox}^{\mathbf D}_{\gamma, g}(\bbb)-\bbb\right)-\left(\operatorname{Prox}^{\mathbf D}_{\gamma, g}\left(\bbb^{\prime}\right)-\bbb^{\prime}\right)\right\|_{\mathbf D}^{2} \leq\left\|\bbb-\bbb^{\prime}\right\|_{\mathbf D}^{2}$.
        \item[(iii)] $\sup _{\gamma \in(0,c_1 / L]} \sup _{\bbb \in \mathcal{B}_1} \gamma^{-1}\left\|\operatorname{Prox}_{\gamma, g}^{\mathbf D}(\bbb)-\bbb\right\|<\infty$.
    \end{itemize}


    \item[H8.] For a measurable function $V: \Xi \rightarrow [1,+\infty),$ a signed measure $\mu$ on the $\sigma$-field of $\Xi,$ and a function $f: \Xi\rightarrow \mathbb{R},$ define
    \begin{equation*}
        |f|_{V} \stackrel{\text { def }}{=} \sup _{\boldsymbol\xi \in \Xi} \frac{|f(\boldsymbol\xi)|}{V(\boldsymbol\xi)}, \quad\|\mu\|_{V} \stackrel{\text { def }}{=} \sup _{f,|f|_{V} \leq 1}\left|\int f \mathrm{d} \mu\right|.
    \end{equation*}
    There exist $\lambda\in(0,1), b<\infty, m\geq 4$ and a measurable function $W: \Xi\rightarrow [1,+\infty)$ such that
    \begin{equation*}
        \sup _{\bbb \in \mathcal{B}_1}\left|\mathbf G_{\bbb}\right|_{W}<\infty, \quad \sup _{\bbb \in \mathcal{B}_1} P_{\bbb} W^{m} \leq \lambda W^{m}+b,
    \end{equation*}
    where $\mathbf G_{\bbb}(\boldsymbol \xi)= \partial H(\boldsymbol\xi,\bbb)/\partial \bbb$ and  $P_{\bbb}$ is the Markov kernel defined in condition H4.
    In addition, for any $\ell\in (0,m], $ there exists $C<\infty, \rho \in (0,1)$ such that for any $\boldsymbol\xi\in\Xi,$
    \begin{equation*}
        \sup _{\bbb \in \mathcal{B}_1}\left\|P_{\bbb}^{n}(\boldsymbol\xi, \cdot)-\pi_{\bbb}\right\|_{W^{\ell}} \leq C \rho^{n} W^{\ell}(\boldsymbol\xi).
    \end{equation*}
    \item[H9.] There exists a constant $C$ such that for any $\bbb,\bbb^{\prime} \in \mathcal{B}_1, $
        \begin{equation*}
            \left|\mathbf G_{\bbb}-\mathbf G_{\bbb^{\prime}}\right|_{W}+\sup _{\boldsymbol\xi\in\Xi} \frac{\left\|P_{\bbb}(\boldsymbol\xi, \cdot)-P_{\bbb^{\prime}}(\boldsymbol\xi, \cdot)\right\|_{W}}{W(\boldsymbol\xi)}+\left\|\pi_{\bbb}-\pi_{\bbb^{\prime}}\right\|_{W} \leq C\left\|\bbb-\bbb^{\prime}\right\|.
        \end{equation*}
\end{itemize}

We provide a few remarks on conditions H6-H9, which are needed for establishing the convergence speed in addition to conditions H1-H5.  {Condition H6 requires that the smooth part of the objective function is  strongly convex and its derivative is Lipschitz continuous in a small neighborhood of $\bbb_*$. Specifically, $h(\bbb)$ being strongly convex in $\mathcal{B}_1$ means that there exists a positive constant $C$, such that $(\nabla h(\bbb)-\nabla h(\bbb'))^\top (\bbb - \bbb') \geq C\Vert \bbb - \bbb'\Vert^2$, for any $\bbb$ and $\bbb' \in \mathcal{B}_1$.}
Condition H7 imposes some requirements on the non-smooth part of the objective function, with regard to the proximal operator.  {As verified in Lemma~C.1, H7 holds when $g$ is a generalized function that indicates constraints or when
$g$ is locally Lipschitz continuous and convex that holds when $g$ is a $L_1$ regularization function. Thus,
 H7 holds for the examples we consider in Sections 2.2 and 2.3.}
Conditions H8 and H9 imposes mild regularity conditions on the stochastic gradient in a local neighborhood of $\bbb_*$, especially when the stochastic gradients are generated by a Markov kernel. These conditions are used to control the bias caused by MCMC sampling. H8 is essentially a
uniform-in-$\bbb$ ergodic condition and H9 is a local Lipschitzian condition on the Markov kernel. These regularity conditions
are commonly adopted in the stochastic approximation literature \citep[]{benveniste1990adaptive,andrieu2005stability,fort2016convergence}, and
have been shown to hold
for general families of MCMC kernels including Metropolis-Hastings and Gibbs samplers (\citealp{andrieuErgodicityPropertiesAdaptive2006}; \citealp[][]{fortConvergenceAdaptiveInteracting2011}; \citealp[][]{schmidtConvergenceRatesInexact2011a}).


\begin{theorem}\label{thm:thm2}
Suppose that H1-H9 hold. Then there exists a constant $C$, such that for the Polyak-Ruppert averaging sequence $\bar\bbb_{n} = \frac{1}{n}\sum_{t=1}^n \bbb^{(t)}$ from Algorithm 1,
\begin{equation}\label{eq:speed}
\mathbb{E}\Vert\bar\bbb_n - \bbb_\star\Vert^2 \leq C n^{-\frac{1}{2} + \varepsilon}.
\end{equation}
Note that the expectation is taken with respect to $\boldsymbol\xi^{(1)},\ldots,\boldsymbol\xi^{(t)}$ given $\bbb^{(0)}$ and $\boldsymbol\xi^{(0)}.$

\end{theorem}

We now provide a few remarks regarding the convergence speed \eqref{eq:speed}.
First, the small positive constant $\varepsilon$ comes from
the requirement on step size that  $\sum_{t=1}^\infty \gamma_t^2 < \infty$  
in H5. Since $\sum_{t=1}^\infty \gamma_t^2 < \infty$ is satisfied when $\gamma_t = \mu t^{-\frac{1}{2} - \varepsilon}$, for any $\mu, \varepsilon > 0$, the convergence speed of $\mathbb{E}\Vert\bar\bbb_n - \bbb_\star\Vert^2$ can be arbitrarily close to $O(n^{-\frac{1}{2}})$ by choosing an arbitrarily small $\varepsilon$. Second, this $\varepsilon$ might be an artifact due to our proof strategy to overcome the non-convexity of the problem. In fact, if the objective function is convex, similar to \cite{atchade2017perturbed}, we can choose $\epsilon = 0$ and then
prove under similar conditions that
$\mathbb{E}\Vert\bar\bbb_n - \bbb_\star\Vert^2 \leq C n^{-\frac{1}{2}}$.
Lastly, it is well-known that for non-smooth convex optimization, the minimax optimal convergence rate is $O(n^{-\frac{1}{2}})$; see Chapter 3, \cite{nesterov2004introductory}.  In this sense, our algorithm is almost minimax optimal, when $\varepsilon$ is very close to zero.
It is well-known that Polyak-Ruppert averaging typically improves the convergence speed of a slowly convergent sequence \citep{ruppert1988efficient,bonnabel2013stochastic}.

\section{Simulation Study}\label{sec:simulation}

\subsection{Study I: Confirmatory IFA}
 \begin{table}
   \centering
   \begin{tabular}{l|cccc}
     \hline
     Estimator & Step size & Averaging &Quasi-Newton  & MCMC \\
     \hline
     USP   &   $\gamma_t = t^{-0.51}$ & Yes  &Yes&Gibbs\\
     USP-PPG   &  $\gamma_t = t^{-0.51}$ & Yes  &No&Gibbs \\
     USP-RM1   &  $\gamma_t = t^{-0.51}$ & No  &Yes &Gibbs \\
     USP-RM2   &  $\gamma_t = t^{-1}$    & No  &Yes &Gibbs\\
     StEM      &NA     &Yes & NA &Gibbs\\
     \hline
   \end{tabular}
   \caption{Comparison of five stochastic algorithms.}\label{tab:methods}
 \end{table}

In the first study, we compare the performance of four variants of the proposed method and the stochastic EM (StEM) algorithm.
The five methods, including their abbreviations are given in Table~\ref{tab:methods}. For a fair comparison, the same Gibbs sampling method is used. We further explain the differences below.
\begin{enumerate}
  \item USP is the method that we recommend. It has a step size $\gamma_t$ close to $t^{-1/2}$, applies Polyak-Ruppert averaging, and uses a quasi-Newton update in the proximal step.
  \item The USP-PPG method is the perturbed proximal gradient method that is implemented the same as the USP method except that  $c_1 = c_2$ so that it does not involve a quasi-Newton update. $c_1$ is set to be 1 without tuning in this study.
  \item The USP-RM1 method is implemented the same as the USP method, except that $\bbb^{(n)}$ from the last iteration is taken as the estimator instead of applying Polyak-Ruppert averaging. This method is very similar to a Robbins-Monro algorithm, except for the update of parameters $\mathbf B$ for the covariance matrix where constraints involve.
  \item The USP-RM2 method is the same as USP-RM1, except that we set the step size $\gamma_t = 1/t$ which is the asymptotic optimal step size for the Robbins-Monro algorithm \citep{fabian1968asymptotic}.
  \item The implementation of the StEM algorithm is the same as USP, except for the proximal step. Instead of making stochastic gradient update, StEM obtains $\bbb^{(t)}$ by completely solving an optimization problem
        $$\bbb^{(t)} = \argmax_{\bbb \in \mathcal B} H(\boldsymbol\xi^{(t)}, \bbb) + g(\bbb).$$
  In our implementation, this optimization problem is solved by making the quasi-Newton proximal update \eqref{eq:prox} iteratively until convergence.
\end{enumerate}

We consider a confirmatory IFA setting with only two factors (i.e., $K=2$), so that an EM algorithm with sufficient numbers of quadrature points and EM steps can be used to obtain a more accurate approximation of $\hat \bbb$ that will be used as the standard when comparing the five methods.  {We emphasize that it is important to compare the convergence speed of difference algorithms based on $\hat \bbb$ rather than the true model parameters. This is because, under suitable conditions, these algorithms converge to $\hat \bbb$ rather than the true model parameters. If we compare the algorithms based on the true model parameters, the difference in the convergence speed cannot be observed clearly, as the statistical error  (i.e., the difference between $\hat \bbb$ and the true model parameters) tends to dominate the computational errors (i.e., the difference between $\hat \bbb$ and the results given by the stochastic algorithms).}

More precisely, we consider sample size $N = 1000$ and the number of items $J = 20$. The design matrix $\mathbf  Q$ is specified by the assumptions that items 1 through 5 only measure the first factor, items 6 through 10 only measure the second factor, and items 11 through 20 measure both.
The intercept parameters $d_j$ are drawn i.i.d. from the standard normal distribution, and the non-zero loading parameters are drawn i.i.d. from a uniform distribution over the interval $(0.5, 1.5)$. The variances of the two factors are set to be 1 and the covariance is set to be 0.4.
Under these parameters, 100 independent datasets are generated, based on which the five methods are compared. To ensure a fair comparison, the true parameters are used as the starting point for all the methods.  In addition, 1000 iterations are run (i.e., $n = 1000$) for each method, instead of using an adaptive stopping criterion. For USP, USP-PPG, and StEM, the burn-in size $\varpi$ is chosen to be 500.
All algorithms are implemented in {\CC} and run on the same platform\footnote{CPU: 2.6 GHz 6-Core Intel Core i7; RAM: 16 GB 2400 MHz DDR4.} using a single core.

The results regarding the accuracy of the proposed methods are given in Figures~\ref{fig:1} and \ref{fig:2} that are based on the following  performance metrics. Specifically, for the intercept parameters $d_j$, the following mean squared error (MSE) is calculated for each simulated dataset and each method,
$$\frac{1}{J}\sum_{j=1}^{J} \left(\tilde d_j - \hat d_j\right)^2,$$
where $\hat d_j$, which is treated as the global optimum, is obtained by an EM algorithm with 31 Gaussian-Hermite quadrature points per dimension, and
$\tilde d_j$ is given by one of the five stochastic methods after 1000 iterations. Similarly, the MSEs for the loading parameters and for the correlation $\sigma_{12}$ between the factors are calculated, where the MSE for the loading parameters is calculated for the unrestricted ones, i.e.,
$$\frac{\sum_{j,k} 1_{\{q_{jk}\neq 0\}} (\tilde a_{jk} - \hat a_{jk})^2 }{\sum_{j,k} 1_{\{q_{jk}\neq 0\}} }.$$
Again,  $\hat a_{jk}$ is given by the EM algorithm, and $\tilde a_{jk}$ is given by one of the five methods.

Figure~\ref{fig:1} compares the accuracy of all the five methods. As we can see, the USP, USP-RM1, and StEM methods have much smaller MSEs than the
USP-PPG and USP-USP-RM2 methods. 
Since the USP-PPG method only differs from the USP method by whether using a quasi-Newton update, the inferior performance of USP-PPG implies the importance of the second-order information in the stochastic proximal gradient update.
As the USP-RM2 method only differs from USP-RM1 by their step sizes, the inferior performance of USP-RM2 is mainly due to the use of short step size.

In Figure~\ref{fig:2}, we zoom in to further compare the USP, USP-RM1, and StEM methods. First, we see that the USP method performs the best among the three, for all the parameters. As the
USP-RM1 method is the same as the USP method except for not applying Polyak-Ruppert averaging, this result suggests that averaging does improve accuracy.  Moreover, the USP method and the StEM method only differ by the way the parameters are updated, where the USP method takes a quasi-Newton proximal update, while the StEM method completely solves an optimization problem. It is likely that the way parameters are updated in the USP method yields more smoothing (i.e., averaging) than the StEM, which leads to the outperformance of the USP method.

\begin{figure}
\centering
\begin{subfigure}[b]{.32\textwidth}
  \centering
  \includegraphics[width=1\linewidth]{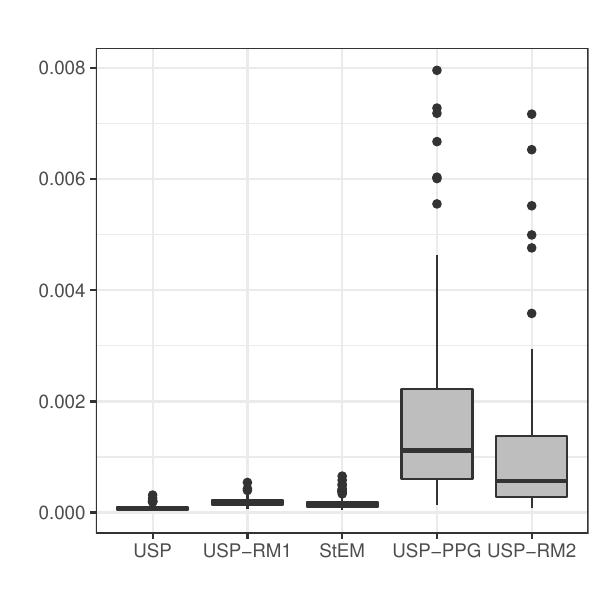}
  \caption{MSE for unrestricted loading parameters.}
\end{subfigure}%
\hspace{\fill}
\begin{subfigure}[b]{.32\textwidth}
  \centering
  \includegraphics[width=1\linewidth]{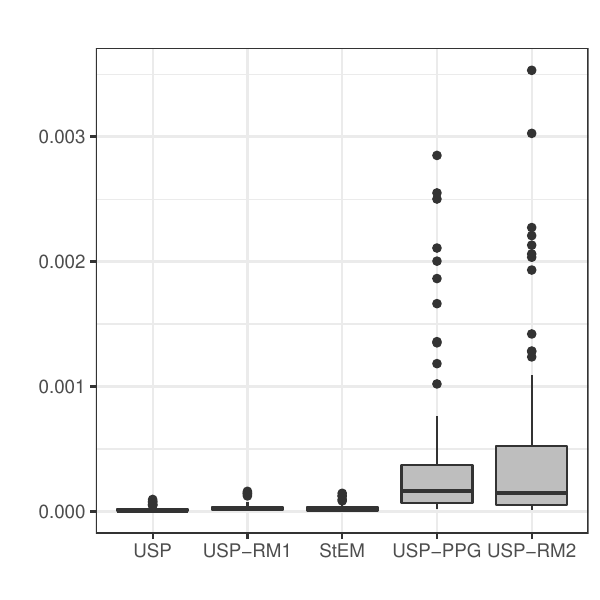}
  \caption{MSE for intercept parameters.}
\end{subfigure}
\hspace{\fill}
\begin{subfigure}[b]{.32\textwidth}
\centering
\includegraphics[width=1\linewidth]{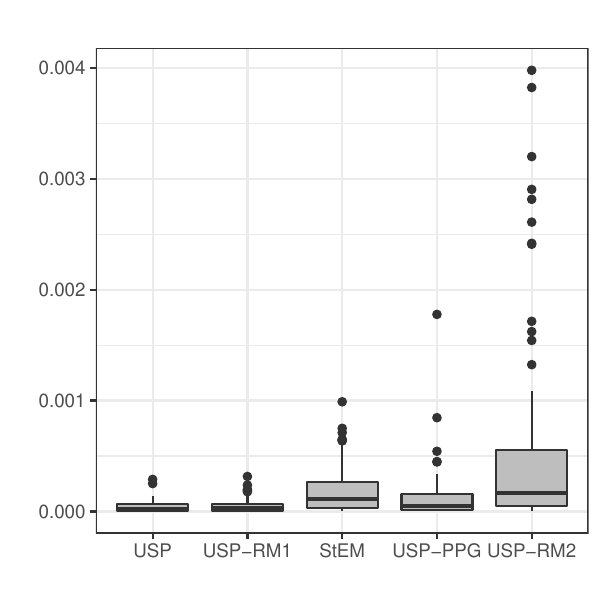}
\caption{MSE for correlation parameter $\sigma_{12}$.}
\end{subfigure}
    \caption{The boxplot of mean squared errors for estimated parameters from the five methods. }\label{fig:1}
\end{figure}

\begin{figure}
\centering
\begin{subfigure}[b]{.32\textwidth}
  \centering
  \includegraphics[width=1\linewidth]{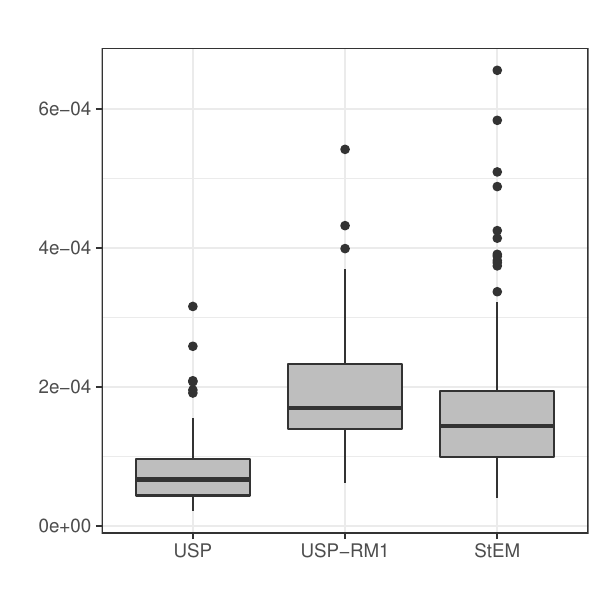}
  \caption{MSE for unrestricted loading parameters. }
\end{subfigure}%
\hspace{\fill}
\begin{subfigure}[b]{.32\textwidth}
  \centering
  \includegraphics[width=1\linewidth]{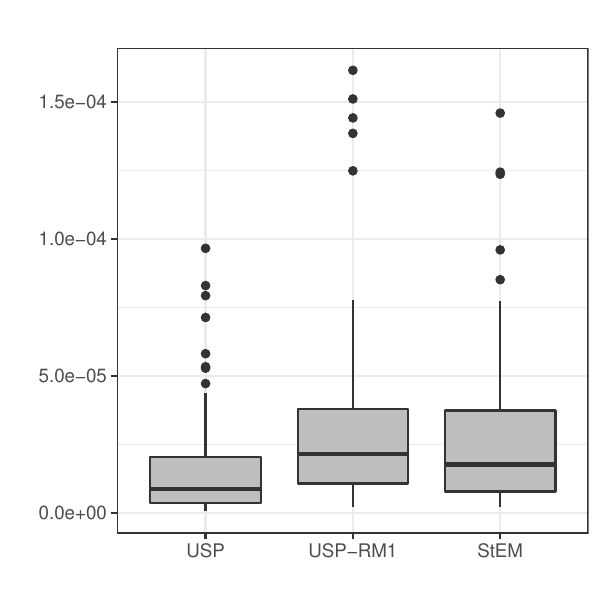}
  \caption{MSE for intercept parameters.}
\end{subfigure}
\hspace{\fill}
\begin{subfigure}[b]{.32\textwidth}
\centering
\includegraphics[width=1\linewidth]{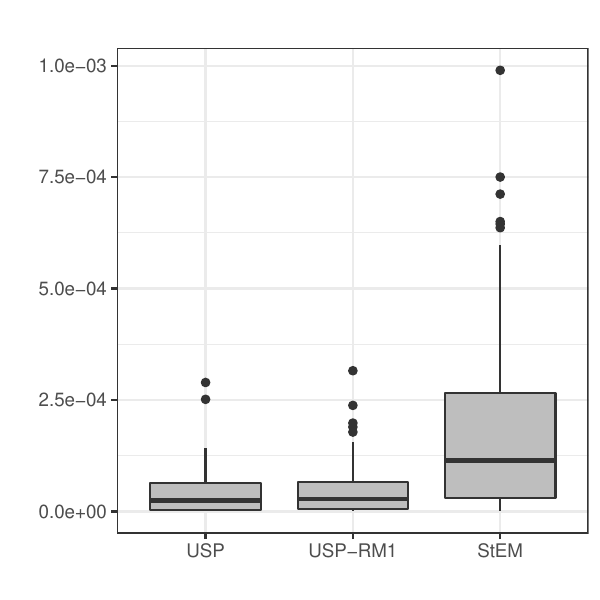}
\caption{MSE for correlation parameter $\sigma_{12}$.}
\end{subfigure}
    \caption{The boxplot of mean squared errors for estimated parameters  from `USP', `USP-RM1', and `StEM' method. }\label{fig:2}
\end{figure}

\begin{table}
\centering
\begin{tabular}{r|ccccc}
\hline
  Elapsed time & USP & USP-RM1 & StEM & USP-PPG & USP-RM2 \\
 \hline
25\% quantile  & 12.2 & 12.2 & 20.3 & 12.2 & 12.1\\
median         & 12.3 & 12.3 & 20.4 & 12.3 & 12.2\\
75\% quantile  & 12.3 & 12.4 & 20.5 & 12.3 & 12.3\\
\hline
\end{tabular}
\caption{The elapsed time (seconds) for the five methods in confirmatory IFA.}\label{tab:time_simu1}
\end{table}

On the computational efficiency, we show in Table~\ref{tab:time_simu1} the elapsed time for the five methods. `USP', `USP-RM1', `USP-PPG', and `USP-RM2' share similar computation time since their floating point operations per iteration are at the same level. `StEM' is most time consuming because an inner loop of optimization is involved in each iteration. {In summary, the proposed USP algorithm is computationally the most efficient among the five algorithms, in the sense that it achieves the highest accuracy (see Figures~\ref{fig:1} and \ref{fig:2}), within a similar or smaller amount of time (see Table~\ref{tab:time_simu1}).}


\subsection{Study II: Exploratory IFA by Regularization}
In the second study, we apply the proposed method to regularized estimation for exploratory IFA as
discussed in Section~\ref{subsec:example1}. We consider
increasing sample size $N=1000, 2000, 4000,$ eighty items and five correlated latent factors (i.e., $J=80,K=5$). The true loading matrix is sparse, where the items each factor loads on are given in Table~\ref{tab:simu2_qmat}. Similar to Study I, the intercept parameters $d_j$ are drawn i.i.d. from the standard normal distribution, and the non-zero loading parameters $a_{jk}$ are drawn i.i.d. from a uniform distribution over the interval (0.5, 1.5). The elements of covariance matrix $\Sigma=(\sigma_{k,k})_{5\times 5}$ are set to be $\sigma_{k,k'}=1,$ for $k=k'$ and $\sigma_{k,k'}=0.4$ for $k\neq k'.$

\begin{table}
\small
\centering
\begin{tabular}{c|l}
\hline
  Factor & Items \\
 \hline
1 &  1-10, 51, 52, 54, 57, 61, 62, 64, 67, 71, 72, 73, 75, 76, 78\\
2 & 11-20, 51, 53, 55, 58, 61, 63, 65, 68, 71, 72, 74, 75, 77, 79\\
3 & 21-30, 52, 53, 56, 59, 62, 63, 66, 69, 71, 73, 74, 76, 77, 80\\
4 & 31-40, 54, 55, 56, 60, 64, 65, 66, 70, 72, 73, 74, 78, 79, 80\\
5 & 41-50, 57, 58, 59, 60, 67, 68, 69, 70, 75, 76, 77, 78, 79, 80\\
\hline
\end{tabular}
\caption{The sparse loading structure in the data generation IFA model.}\label{tab:simu2_qmat}
\end{table}


For each sample size, 50 independent datasets are generated. In the proposed algorithm, we adopt a burn-in size $\varpi=50$ and stop based on the criterion discussed in Section~\ref{sec:algorithm}, where the stopping threshold is set to be $10^{-3}$. A decreasing penalty parameter $\lambda_N = \sqrt{\log J /N}$ is used to ensure estimation consistency \citep[Chapter 6,][]{buhlmannStatisticsHighDimensionalData2011}. Other implementation details can be found in Section~\ref{subsec:exmp1}.
The algorithm in this example is implemented in {\CC} and is run on the same platform as in Study I.  {Although  a regularized EM algorithm \citep{sun2016latent} can also solve this problem, it suffers from a very high computational cost. Due to the five-dimensional numerical integrals involved, it takes a few hours to fit one dataset. We thus do not consider it here.}

We focus on the accuracy in the estimation of the loading matrix $\mathbf A = (a_{jk})_{J\times K}$.
Note that although the rotational indeterminacy issue is resolved in this regularized estimator, the loading matrix can still only be identified up to column swapping. That is, two estimates of the loading matrix
have the same objective function value, if one can be obtained by swapping the columns of the other. The following mean-squared-error measure is used that takes into account column swapping of the loading matrix
\begin{equation}\label{def:distance}
    \min_{\mathbf A'\in \mathcal P(\tilde{\mathbf A})}\left\{\frac{1}{JK}\Vert \mathbf A' - \mathbf A\Vert_F^2\right\},
\end{equation}
where  $\Vert\cdot\Vert_F$ is the Frobenius norm, $\mathbf A$ is the true loading matrix, $\tilde{\mathbf A}$ is the output of Algorithm~\ref{alg:alg1}, and $\mathcal P(\tilde{\mathbf A})$ denotes the set of $J\times K$ matrices that can be obtained by swapping the columns of $\tilde{\mathbf A}$.

Results are given in Tables~\ref{tab:mse_a_penmirt} and \ref{tab:time_penmirt}.
In Table~\ref{tab:mse_a_penmirt}, we see that the MSE for the loading matrix is quite small and decreases as the sample size grows, suggesting that consistency of the regularized estimator.
In Table~\ref{tab:time_penmirt}, the quantiles of time consumption under different sample sizes are given, which suggests the computational efficiency of the proposed method.

\begin{table}
\centering
\begin{tabular}{r|rrr}
\hline
  MSE of $\tilde{\mathbf A}$ & $N$=1000 & $N$=2000 & $N$=4000 \\
 \hline
25\% quantile  & 0.032 & 0.026 & 0.017\\
median         & 0.034 & 0.027 & 0.018\\
75\% quantile  & 0.036 & 0.027 & 0.019\\
\hline
\end{tabular}
\caption{The mean squared errors for estimated loading parameters in exploratory IFA with $L_1$ regularization.}\label{tab:mse_a_penmirt}
\end{table}

\begin{table}
\centering
\begin{tabular}{r|rrr}
\hline
  Elapsed time & $N$=1000 & $N$=2000 & $N$=4000 \\
 \hline
25\% quantile  &  9.2 & 14.8 & 25.8\\
median         &  9.7 & 15.1 & 26.6\\
75\% quantile  & 10.4 & 15.7 & 27.6\\
\hline
\end{tabular}
\caption{The elapsed time (seconds) for exploratory IFA with $L_1$ regularization.}\label{tab:time_penmirt}
\end{table}

\subsection{Study III: Restricted LCA}\label{subsec:simlca}

In this study, we apply the proposed method to the estimation of a restricted latent class model as discussed
in Section~\ref{subsec:example2}, where the optimization involves complex inequality constraints. Specifically, data are from a Deterministic Input, Noisy `And' gate (DINA) model \citep{junker2001cognitive} that is a special restricted latent class model.
Note that the DINA assumptions are only used in the data generation. We solve optimization~\eqref{eq:dcm} which is based on a general restricted latent class model considered in \cite{xu2017identifiability}
instead of the DINA model, mimicing the practical situation when the parametric form is unknown.

We consider a test consisting of twenty items (i.e., $J = 20$) that measure four binary attributes (i.e., $K=4$). Three sample sizes are considered, including $N=1000, 2000$, and 4000. The design matrix $\mathbf Q$ is given in Table~\ref{tab:simu3_qmatrix}. In addition, the guessing and slipping parameters $s_j$ and $g_j$
of the DINA model are drawn i.i.d. from a uniform distribution over the interval (0.05, 0.2), which gives the values of
$\theta_{j, \boldsymbol\alpha}$. That is,
\begin{equation*}
\theta_{j, \boldsymbol\alpha} =\left\{ \begin{array}{cl}
                                  \log( (1-s_j)/s_j), & \mbox{~if~} \boldsymbol\alpha \succeq \mathbf q_{j}, \\
                                  \log( g_j/(1-g_j)), & \mbox{~otherwise.}
                                \end{array}\right.
\end{equation*}
Finally, we let $\nu_{\boldsymbol \alpha} = 0,$ for all $\boldsymbol \alpha \in \{0, 1\}^K$, so that $\mathbb{P}(\boldsymbol\xi = \boldsymbol\alpha) = 1/2^K$.
According to the results in \cite{xu2017identifiability}, the model parameters are indentifiable, given the $\mathbf Q$-matrix in Table~\ref{tab:simu3_qmatrix}.

\begin{table}
\small
\centering
\begin{tabular}{c|cccccccccccccccccccc}
\hline
 Attribute & \multicolumn{20}{c}{ Items } \\\cline{2-21}
           & 1  & 2  & 3  & 4  & 5  & 6  & 7  & 8  & 9  & 10 & 11 & 12 & 13 & 14 & 15 & 16 & 17 & 18 & 19 & 20 \\
 \hline
$\alpha_1$ & 1 & 0 & 0 & 0 & 1 & 0 & 0 & 0 & 1 & 0 & 0 & 0 & 1 & 1 & 0 & 1 & 1 & 1 & 1 & 0 \\
$\alpha_2$ & 0 & 1 & 0 & 0 & 0 & 1 & 0 & 0 & 0 & 1 & 0 & 0 & 1 & 0 & 1 & 0 & 1 & 1 & 0 & 1 \\
$\alpha_3$ & 0 & 0 & 1 & 0 & 0 & 0 & 1 & 0 & 0 & 0 & 1 & 0 & 0 & 1 & 1 & 0 & 1 & 0 & 1 & 1 \\
$\alpha_4$ & 0 & 0 & 0 & 1 & 0 & 0 & 0 & 1 & 0 & 0 & 0 & 1 & 0 & 0 & 0 & 1 & 0 & 1 & 1 & 1\\
\hline
\end{tabular}
\caption{The design matrix $\mathbf Q$ for the restricted LCA model.}\label{tab:simu3_qmatrix}
\end{table}

For each sample size, 50 independent datasets are generated. The proposed algorithm adopts a burn-in size $\varpi=50$ and stops based on the criterion discussed in Section~\ref{sec:algorithm}, where the stopping threshold is set to be $10^{-3}$.
Other implementation details can be found in Section~\ref{subsec:exmp2}.
The following metrics are used to evaluate the estimation accuracy. For item parameters $\theta_{j, \boldsymbol\alpha}$, the MSE is calculated as
\[
    \frac{1}{J\times 2^K}\sum_{j=1}^J\sum_{\boldsymbol\alpha\in\{0,1\}^K}\left(\tilde\theta_{j,\boldsymbol\alpha}-\theta_{j,\boldsymbol\alpha}\right)^2.
\]
For structural parameters $\nu_{\boldsymbol\alpha}$, the MSE is calculated as
\[
    \frac{1}{2^K-1}\sum_{\boldsymbol\alpha\in\{0,1\}^K,\ \boldsymbol\alpha\neq \boldsymbol 0}\left(\tilde\nu_{\boldsymbol\alpha} - \nu_{\boldsymbol\alpha}\right)^2.
\]
Our results are given in Tables~\ref{tab:mse_lca} and \ref{tab:kl_lca}.
 As we can see, the estimation becomes more accurate  as the sample size increases for both sets of parameters. It confirms that the current model is identifiable as suggested by \cite{xu2017identifiability} and thus can be consistently estimated.


\begin{table}
\centering
\begin{tabular}{r|rrr}
\hline
  MSE of $\tilde{\boldsymbol\theta}$ & $N$=1000 & $N$=2000 & $N$=4000 \\
 \hline
25\% quantile  & 0.150 & 0.062 & 0.028\\
median         & 0.182 & 0.070 & 0.031\\
75\% quantile  & 0.252 & 0.077 & 0.033\\
\hline
\end{tabular}
\caption{The MSE for item parameters $\theta_{j, \boldsymbol\alpha}$  in the restricted latent class model.}\label{tab:mse_lca}
\end{table}

\begin{table}
\centering
\begin{tabular}{r|rrr}
\hline
  MSE of $\tilde{\boldsymbol \nu}_{\boldsymbol\alpha}$ & $N$=1000 & $N$=2000 & $N$=4000 \\
 \hline
25\% quantile  & 0.028 & 0.012 & 0.005\\
median         & 0.045 & 0.018 & 0.007\\
75\% quantile  & 0.085 & 0.028 & 0.009\\
\hline
\end{tabular}
\caption{The MSE for structural parameters $\nu_{\boldsymbol\alpha}$ in the restricted latent class model.}\label{tab:kl_lca}
\end{table}

\section{Concluding Remarks}\label{sec:conclusion}

In this paper, a unified stochastic proximal optimization framework is proposed for the computation of latent variable model estimation.
This framework is very general that applies to a wide range of estimators for almost all commonly used latent variable models.
Comparing with existing stochastic optimization methods, the proposed method not only solves a wider range of problems including regularized and constrained estimators,  but also is computationally more efficient.
Theoretical properties of the proposed method are established. These results suggest that the
convergence speed of the proposed method is almost optimal in
the minimax sense.


The power of the proposed method is shown via three examples, including confirmatory IFA, exploratory IFA by regularized estimation, and restricted latent class analysis. Specifically, the proposed method is compared with several stochastic optimization algorithms, including a stochastic-EM algorithm and a Robbin-Monro algorithm with MCMC sampling, in the simulation study of confirmatory IFA, where there is no complex constraint or penalty. Using the same starting point and the same number of iterations, the proposed one is always more accurate than its competitors.  The simulation studies on exploratory IFA and restricted latent class analysis further show the power of the proposed method for handling optimization problems with non-smooth penalties and complex inequality constraints.

 {The implementation of the proposed algorithm involves several tuning parameters. First, we need to choose a step size $\gamma_t$. Our theoretical results suggest that $\gamma_t =t^{-0.5 - \epsilon}$ for any $\epsilon\in (0,0.5]$, and a smaller $\epsilon$ leads to faster convergence.  In practice, we suggest to set $\gamma_t =t^{-0.51}$ that performs well in all our simulations. This choice of step size is very different from the choice of $\gamma_t =t^{-1}$ in the MCMC stochastic approximation algorithms. Second, a burn-in size $\varpi$ is needed. The burn-in in the proposed algorithm is similar to the burn-in in MCMC algorithms. It does not affect the asymptotic convergence of the algorithm but improves the finite sample performance. In practice, the burn-in size can be decided similarly as in MCMC algorithms by
monitoring the parameter updates using trace plots. Third, two positive constraints $c_1$ and $c_2$ are needed to regularize the second-order matrix in the scaled proximal update. Depending on the scale of each particular problem, we suggest to choose $c_1$ to be sufficiently small and $c_2$ to be sufficiently large.
It is found that the performance of our algorithm is not sensitive to their choices.
Finally, a stopping criterion is needed. 
We suggest to stop the iterative update by monitoring a window of successive differences in parameter updates.
}

 {The proposed framework may be improved from several aspects that are left for future investigation. First, the sampling strategy in the stochastic step needs further investigation. Although in theory any reasonable MCMC sampler can yield the convergence of the algorithm, a good sampler will lead to superior finite sample performance. More sophisticated MCMC algorithms need to be investigated regarding their performance under the proposed framework. Second, methods for parallel and distributed computing need to be developed. As we can see, many steps of Algorithm~\ref{alg:alg1} can be performed independently. This enables us to design parallel and/or distributed computing systems for solving large-scale and/or distributed versions of latent variable model estimation problems (e.g., fitting models for assessment data from online learning platforms and large-scale mental health records). Finally, the performance of the proposed method under other latent variable models needs to be investigated. For example, the proposed method can also be applied to  latent stochastic process models \citep[e.g.,][]{chow2016fitting,chen2020latent} that are useful for analyzing intensive longitudinal data. These models bring additional challenges, as stochastic processes need to be sampled in the stochastic step of our algorithm.   }

In summary, the proposed method is computationally efficient, theoretically solid, and applicable to a broad range of latent variable model inference problems. Like the EM algorithm as the standard tool for low-dimensional latent variable models, we believe that the proposed method may potentially serve as the standard approach to the estimation of high-dimensional latent variable models.




\appendix
In this supplement, we provide proofs of theoretical results in the main manuscript.
We define some notations:
\begin{itemize}
	\item $F(\bbb) = h(\bbb) + g(\bbb)$
    \item $\partial f(\boldsymbol x)=\left\{\boldsymbol z \in \mathbb{R}^{p}: f(\boldsymbol y) \geq f(\boldsymbol x)+ \boldsymbol z^\top(\boldsymbol y-\boldsymbol x)+o(\|\boldsymbol y-\boldsymbol x\|) \text { as } \boldsymbol y \rightarrow \boldsymbol x\right\}$
	\item $\mathbf G_{\bbb}(\boldsymbol\xi) = \partial H(\boldsymbol\xi,\bbb)/\partial \bbb$
	\item $\bbb_{\gamma}^+(\boldsymbol\xi) = \underset{\boldsymbol x \in \mathcal{B}}{\operatorname{argmin}}\left\{\boldsymbol G_{\bbb}(\boldsymbol\xi)^\top(\boldsymbol x-\bbb)+\g(\boldsymbol x)+\frac{1}{2 \gamma}\left\|\boldsymbol x-\bbb\right\|^{2}_{\mathbf D}\right\}$
	\item $\boldsymbol{U}_{\gamma}(\boldsymbol\xi;\bbb) = \frac{1}{\gamma}(\bbb - \bbb_{\gamma}^+(\boldsymbol\xi))$
	\item $\mathbb{E}(\cdot\mid \bbb) = \int \cdot\pi_{\bbb}(\boldsymbol\xi)d\boldsymbol\xi,$ $\pi_{\bbb}(\boldsymbol\xi)$ is the posterior density for $\boldsymbol\xi$ given $\boldsymbol y$ and $\bbb$
	\item $\mathcal{F}_{t-1} = \sigma(\bbb^{(0)},\boldsymbol\xi^{(k)},0\leq k\leq t-1)$ is a filtration of $\sigma$-field
	\item $\mathcal{C}(\mathbb{R}^+,\mathbb{R}^p)$ denotes the continuous functions from $\mathbb{R}^+$ to $\mathbb{R}^p$
	\item $\operatorname{Prox}_{\gamma, g}^{\mathbf{D}}(\boldsymbol{\beta})=\underset{\mathbf{x} \in \mathbb{R}^{n}}{\arg \min }\left\{g(\mathbf{x})+\frac{1}{2 \gamma}\|\mathbf{x}-\boldsymbol{\beta}\|_{\mathbf{D}}^{2}\right\}$
\end{itemize}

\section{Proof of Lemma 1} 
\label{proof:lemma_1}

Our stochastic updates can be re-formated as
\begin{equation}
    \bbb^{(t)} = \bbb^{(t-1)} - \gamma_t \boldsymbol{U}_{\gamma_t}(\boldsymbol\xi^{(t)};\bbb^{(t-1)}).
\end{equation}
Let $\boldsymbol U(\boldsymbol\xi;\bbb) = \partial H(\boldsymbol\xi,\bbb)/\partial\bbb+\partial g(\bbb)$ and $\boldsymbol\epsilon_{\gamma}(\boldsymbol\xi;\bbb) = \boldsymbol{U}_{\gamma}(\boldsymbol\xi;\bbb) - \mathbb{E}[\boldsymbol{U}_{\gamma}(\boldsymbol\xi;\bbb)\mid\bbb]$. By Lemma 7 of \cite{duchiStochasticMethodsComposite2018}, for $\bbb\in\mathcal{B}$ and $\epsilon > 0,$
\begin{equation*}
    \begin{aligned}
        \Vert\boldsymbol{U}_{\gamma}(\boldsymbol\xi;\bbb)\Vert&\leq \Vert \boldsymbol U(\boldsymbol\xi;\bbb)\Vert,\text{ and}\\
        \mathbb{E}[\Vert\boldsymbol\epsilon_{\gamma}(\boldsymbol\xi;\bbb)\Vert^2\mid \bbb]&\leq \mathbb{E}[\Vert\boldsymbol U_\gamma(\boldsymbol\xi;\bbb)]\Vert^2\mid\bbb]\\
        &\leq \mathbb{E}[L^2_{\epsilon}(\bbb ; \boldsymbol\xi)\mid\bbb]\\
        & = L_\epsilon(\bbb)^2,
    \end{aligned}
\end{equation*}
where $L_{\epsilon}(\bbb ; \boldsymbol\xi)=\sup _{\bbb^{\prime} \in \mathcal{B},\left\|\bbb^{\prime}-\bbb\right\| \leq \epsilon}\left\|\boldsymbol U(\boldsymbol\xi;\bbb)\right\|,$ $L_{\epsilon}(\bbb)=\mathbb{E}\left[L_{\epsilon}(\bbb ; \boldsymbol\xi)^{2}\mid\bbb\right]^{\frac{1}{2}}$. And $L_\epsilon(\bbb)<\infty$ for all $\bbb\in\mathcal{B}$ by Lemma 8 of \cite{duchiStochasticMethodsComposite2018}.


So we have
\begin{equation}\label{noise_bound}
    \mathbb{E}[\boldsymbol\epsilon_{\gamma_t}(\boldsymbol\xi^{(t)};\bbb^{(t-1)})\mid \mathcal{F}_{t-1}] = 0,\quad \mathbb{E}[\Vert\boldsymbol\epsilon_{\gamma_t}(\boldsymbol\xi^{(t)};\bbb^{(t-1)})\Vert^2\mid \mathcal{F}_{t-1}] \leq L_\epsilon(\bbb^{(t-1)})^2,
\end{equation}
since $\bbb^{(t-1)}\in \mathcal{F}_{t-1},$ and given $\bbb^{(t-1)},$ $\boldsymbol\xi^{(t)}$ is independent of $\boldsymbol\xi^{(s)},s<t.$ Note that the independence holds true for exact sampling; For MCMC sampling, independence can also be achieved for any precision after applying `thinning' procedure.

Further since $\mathcal{B}$ is compact, there is a random variable $B$ which is finite with probability 1, such that for $t\in\mathbb{N},$ $\Vert \bbb^{(t)}\Vert\leq B.$ Together with step size condition in H5, we have
\[
    \sum_{t=1}^{\infty}\mathbb{E}\left[\gamma_t^2\Vert\boldsymbol\epsilon_{\gamma_t}(\boldsymbol\xi^{(t)};\bbb^{(t-1)})\Vert^2\mid\mathcal{F}_{t-1}\right]\leq\sum_{t=1}^{\infty}\gamma_t^2\sup_{\Vert\bbb\Vert\leq  B,\bbb\in\mathcal{B}} L_{\epsilon}(\bbb)^2\leq\infty.
\]
Thus $\gamma_t\boldsymbol{\epsilon}_{\gamma_t}(\boldsymbol\xi^{(t)},\bbb^{(t-1)})$ is a $l_2$-summable martingale difference sequence adpated to $\mathcal{F}_{t-1}.$ By standard martingale convergence result \citep[e.g., Thm. 5.3.33, ][]{demboProbabilityTheorySTAT310},  we have with probability 1, $\lim_n \sum_{t=1}^n \gamma_t \boldsymbol{\epsilon}_{\gamma_t}(\boldsymbol\xi^{(t)},\bbb^{(t-1)})$ exist and is finite.


\section{Proof of Theorem 1}
In Theorem 1, we establish the convergence of $\bbb^{(t)}$ to a stationary point $\bbb_\infty\in \mathcal{B}^*$ using differential inclusion techniques in \cite{duchiStochasticMethodsComposite2018}. The proposed method can be viewed as a special case of the general stochastic method discussed in \cite{duchiStochasticMethodsComposite2018} with a few differences.

With additional assumptions, a similar convergence result can be derived. In what follows, we first show the linear interpolation process of our stochastic updates is asymptotically equivalent to a differential inclusion, by verifying that conditions of Theorem 2 in \cite{duchiStochasticMethodsComposite2018} hold for our case. Then, cluster points of any trajectory of the limiting differential inclusion are proved to be stationary points. Lastly, the convergence properties of our original sequence can be shown from the functional convergence.


First we define the linear interpolation of the iterates $\bbb^{(k)}$:
\[
	\bbb(t) = \bbb^{(k)} + \frac{t-t_k}{t_{k+1}-t_k}(\bbb^{(k+1)}-\bbb^{(k)})\text{ and } y(t) = \mathbb{E}[\boldsymbol{U}_{\gamma_k}(\boldsymbol\xi^{(k)};\bbb^{(k-1)})\mid\bbb^{(k-1)}]\text{ for } t\in [t_k,t_{k+1}),
\]
and $\bbb^t(\cdot) = \bbb(t+\cdot),t\in\mathbb{R}_+$ be the time-shifted process.

In order to use Theorem 2 of \cite{duchiStochasticMethodsComposite2018}, which is a general functional convergence theorem, conditions (i)-(iv) of Theorem 2 need to be verified for our case. Firstly, the boundness condition (i) holds as $\mathcal{B}$ is compact given H1; Non-summable but square-summable steps size condition (ii) holds given H5; And we have verified (iii), which is the convergence of the summation of the weighted noise sequence, holds by Lemma 1; Lastly, condition (iv) holds similarly in our case for the close-value mapping $-\boldsymbol U(\bbb)-\mathcal{N}_{\mathcal{B}}(\bbb)$ (see Lemma 10 in \cite{duchiStochasticMethodsComposite2018}), where $\boldsymbol U(\bbb) = \nabla h(\bbb)+\partial g(\bbb)$ and $\mathcal{N}_{\mathcal{B}}(\bbb) = \left\{\boldsymbol v\in\mathbb{R}^p:\langle \boldsymbol v, \bbb^\prime - \bbb\rangle, \text{for all } \bbb^\prime\in \mathcal{B}\right\}$ is the normal cone for $\mathcal{B}$ at $\bbb.$

Based on Theorem 2 and Theorem 3 of \cite{duchiStochasticMethodsComposite2018}, for any sequences $\{\tau_k\}_{k=1}^\infty,$ the function sequence $\bbb^{\tau_k}(\cdot)$ is relatively compact in $\mathcal{C}(\mathbb{R}^+,\mathbb{R}^p)$ and for any $\tau_k\rightarrow\infty,$ any limit point of $\{\bbb^{\tau_k}(\cdot)\}$ in $\mathcal{C}(\mathbb{R}^+,\mathbb{R}^p)$ satisfies
\[
	\bar\bbb(t) = \bar\bbb(0) + \int_0^t y(\tau)d\tau,\text{ where } y(\tau)\in -\boldsymbol U(\bbb(\tau))-\mathcal{N}_{\mathcal{B}}(\bbb(\tau)).
\]
So the sample path of our algorithm is asymptotically equivalent to the differential inclusion
\begin{equation}\label{differentialeq}
	\dot{\bbb} \in -\boldsymbol U(\bbb) - \boldsymbol N_{\mathcal{B}}(\bbb).
\end{equation}
and the converged differential inclusion have uniqueness and convergence properties (see Theorem 4 of \cite{duchiStochasticMethodsComposite2018}).

Finally, according to Theorem 1 of \cite{duchiStochasticMethodsComposite2018}, with probability 1,
\[
	[\underset{t}{\lim\inf} F(\bbb^{(t)}),\underset{t}{\lim\sup} F(\bbb^{(t)})] \subset F(\mathcal{B}^*).
\]
Consequently, given assumption H1, $\mathcal{B}$ is compact and $\mathcal{B}^*$ contains finite points, we have the objective value $F(\bbb^{(t)})$ converges and all cluster points of the sequence $\{\bbb^{(t)}\}$ belong to $\mathcal{B}^*.$

By further assumption that different stationary points in $\mathcal{B}^*$ have different objective values, we have $\bbb^{(t)}$ converges to a stationary point in $\mathcal{B}^*,$ with probability 1.

\section{Proof of Theorem~\ref{thm:thm2}}\label{proof_thm2}
Follow the proofs in Section 6 of \cite{atchade2017perturbed}, we first prove several lemmas, then prove Theorem~\ref{thm:thm2}.

\begin{lemma}\label{lemma:g_bound}
	If $g$ is convex and Lipschitz on $\mathcal{B}_1$ with Lipschitz constant K, or $g=I_{\mathcal{B}}(\cdot).$ For $\bbb,\bbb' \in \mathcal{B}_1,$ any $\gamma>0,$ and diagonal matrix $\mathbf D$ with diagonal entries $\delta_i\in[c_1,c_2], c_2\geq c_1>0,$ the following conditions hold.
    \begin{itemize}
        \item[(i)]
        $g\left(\operatorname{Prox}^{\mathbf D}_{\gamma, g}(\bbb)\right)-g\left(\bbb^{\prime}\right) \leq-\frac{1}{\gamma}\left\langle\operatorname{Prox}^{\mathbf D}_{\gamma, g}(\bbb)-\bbb^{\prime}, \operatorname{Prox}^{\mathbf D}_{\gamma, g}(\bbb)-\bbb\right\rangle_{\mathbf D}$.
        \item[(ii)]
        $\left\|\operatorname{Prox}^{\mathbf D}_{\gamma, g}(\bbb)-\operatorname{Prox}^{\mathbf D}_{\gamma,{} g}\left(\bbb^{\prime}\right)\right\|_{\mathbf D}^{2}+\left\|\left(\operatorname{Prox}^{\mathbf D}_{\gamma, g}(\bbb)-\bbb\right)-\left(\operatorname{Prox}^{\mathbf D}_{\gamma, g}\left(\bbb^{\prime}\right)-\bbb^{\prime}\right)\right\|_{\mathbf D}^{2} \leq\left\|\bbb-\bbb^{\prime}\right\|_{\mathbf D}^{2}$.
        \item[(iii)] $\sup _{\gamma \in(0,c_1 / L]} \sup _{\bbb \in \mathcal{B}_1} \gamma^{-1}\left\|\operatorname{Prox}_{\gamma, g}^{\mathbf D}(\bbb)-\bbb\right\|<\infty$.
    \end{itemize}
\end{lemma}

\begin{proof}[Proof of Lemma~\ref{lemma:g_bound}]\label{proof:lemma_a1}
~
	
	If $g=I_{\mathcal{B}}(\cdot),$ then for $\bbb\in\mathcal{B}_1\subset\mathcal{B},$ $\text{Prox}_{\gamma,g}^{\mathbf{D}}(\bbb) = \bbb,$ so (i)-(iii) hold.

	If $g$ is Lipschitz (thus lower semi-continuous) and convex, given $\bbb, \bbb' \in \mathcal{B}_1,$ $\gamma>0,$

	Let $\ppp = \operatorname{Prox}_{\gamma, g}^{\mathbf{D}}(\boldsymbol{\beta})=\underset{\mathbf{x} \in \mathbb{R}^{n}}{\arg \min }\left\{g(\mathbf{x})+\frac{1}{2 \gamma}\|\mathbf{x}-\boldsymbol{\beta}\|_{\mathbf{D}}^{2}\right\},$ set $\ppp_\alpha = \alpha \bbb' + (1-\alpha)\ppp,$ for $\alpha\in (0,1).$ We have
	\[
	g(\ppp) + \frac{1}{2\gamma}\Vert\ppp-\bbb\Vert_{\mathbf D}^2 \leq g(\ppp_\alpha) + \frac{1}{2\gamma}\Vert\ppp_\alpha-\bbb\Vert_{\mathbf D}^2
	\]
	Due to the convexity of $g$,
	\begin{align*}
	g(\ppp)&\leq \alpha g(\bbb') + (1-\alpha)g(\ppp) + \frac{1}{2\gamma}\Vert\alpha\bbb' - \alpha\ppp + \ppp-\bbb\Vert_{\mathbf D}^2 - \frac{1}{2\gamma}\Vert\ppp-\bbb\Vert_{\mathbf D}^2\\
	&\leq \alpha g(\bbb') + (1-\alpha)g(\ppp) - \frac{\alpha}{\gamma}\langle \ppp-\bbb', \ppp-\bbb\rangle_{\mathbf D} + \frac{\alpha^2}{2\gamma}\Vert\bbb'-\ppp\Vert_{\mathbf D}^2.
	\end{align*}
	So
	\[
	g(\ppp) - g(\bbb')\leq -\frac{1}{\gamma}\langle\ppp-\bbb',\ppp-\bbb\rangle_{\mathbf D} + \frac{\alpha}{2\gamma}\Vert\bbb'-\ppp\Vert_{\mathbf D}^2
	\]
	Let $\alpha \downarrow 0,$ we have the desired inequality (i).

	Further let $\qqq = \operatorname{Prox}_{\gamma, g}^{\mathbf{D}}(\bbb'),$ by (i), we have
	\begin{align*}
	g(\ppp) + \frac{1}{\gamma}\langle\ppp-\qqq,\ppp-\bbb\rangle_{\mathbf D}&\leq g(\qqq)\\
	g(\qqq) + \frac{1}{\gamma}\langle\qqq - \ppp,\qqq-\bbb'\rangle_{\mathbf D} &\leq g(\ppp)
	\end{align*}

	So
	\[
	0 \leq \langle\ppp-\qqq,\bbb-\bbb' -\ppp+\qqq)\rangle_{\mathbf D},
	\]
 	and
 	\begin{align*}
 		\Vert \ppp-\qqq\Vert_{\mathbf D}^2&\leq \langle\ppp-\qqq, \bbb-\bbb'\rangle_{\mathbf D}\\
 		\Vert (\ppp-\bbb) - (\qqq-\bbb')\Vert_{\mathbf D}^2&\leq \langle (\bbb-\ppp)-(\bbb'-\qqq),\bbb-\bbb'\rangle_{\mathbf D}
 	\end{align*}
 	Summation of the above two inequations yeilds (ii).

Given $g$ is proper convex, Lipschitz on $\mathcal{B}_1$ with Lipschitz constant is $K$ and (i), we have
\[
	0\leq \gamma^{-1}\left\|\operatorname{Prox}^{\mathbf D}_{\gamma, g}(\bbb)-\bbb\right\|_{\mathbf D}^2\leq g(\bbb) - g(\operatorname{Prox}^{\mathbf D}_{\gamma, g}(\bbb)) \leq K \left\|\operatorname{Prox}^{\mathbf D}_{\gamma, g}(\bbb)-\bbb\right\|_{\mathbf D}.
\]

Thus (iii) holds.

\end{proof}

\begin{lemma}\label{lemma2}
Assume H7 and $\gamma\in (0,c_1/L],$ for $\bbb,\bbb',\boldsymbol\xi\in \mathcal{B}_1,$
\begin{equation}\label{eq:a3}
	\begin{aligned}
-2 \gamma\left(F\left(\operatorname{Prox}^{\mathbf D}_{\gamma, g}(\bbb)\right)-F\left(\bbb^{\prime}\right)\right) & \geq\left\|\operatorname{Prox}^{\mathbf D}_{\gamma, g}(\bbb)-\bbb^{\prime}\right\|_{\mathbf D}^{2} \\
&+2\left\langle\operatorname{Prox}^{\mathbf D}_{\gamma, g}(\bbb)-\bbb^{\prime}, \boldsymbol\xi-\gamma\mathbf D^{-1} \nabla h(\boldsymbol\xi)-\bbb\right\rangle_{\mathbf D}-\left\|\bbb^{\prime}-\boldsymbol\xi\right\|_{\mathbf D}^{2}
\end{aligned}
\end{equation}
\end{lemma}

\begin{proof}[Proof of Lemma~\ref{lemma2}]
~

Using descent lemma of Lipschitz function $\nabla h,$ for any $\gamma^{-1}\geq L/c_1,$
\begin{equation*}
	h(\ppp)-h\left(\boldsymbol\xi\right) \leq\left\langle\mathbf D^{-1}\nabla h\left(\boldsymbol\xi\right), \ppp-\boldsymbol\xi\right\rangle_{\mathbf D}+\frac{1}{2 \gamma}\left\|\ppp-\boldsymbol\xi\right\|_{\mathbf D}^{2}
\end{equation*}
Since $h$ is convex, so $h(\boldsymbol\xi) + \langle\nabla h(\boldsymbol\xi),\bbb'-\boldsymbol\xi\rangle\leq h(\bbb'),$
\[
	h(\ppp) - h(\bbb') \leq \left\langle\mathbf D^{-1}\nabla h(\boldsymbol\xi), \ppp-\bbb'\right\rangle_{\mathbf D} + \frac{1}{2\gamma}\Vert\ppp-\boldsymbol\xi\Vert^2_{\mathbf D}
\]
And
\[
	g(\ppp)-g(\bbb') \leq -\frac{1}{\gamma}\langle\ppp-\bbb',\ppp-\bbb\rangle_{\mathbf D}
\]
Summation of the above two, we have,
\[
	F(\ppp) - F(\bbb')\leq -\frac{1}{\gamma}\left\langle \ppp-\bbb',\boldsymbol\xi-\gamma\mathbf D^{-1}\nabla h(\boldsymbol\xi) - \bbb\right\rangle_{\mathbf D} + \frac{1}{2\gamma}\Vert\bbb'-\boldsymbol\xi\Vert_{\mathbf D}^2 - \frac{1}{2\gamma}\Vert\ppp-\bbb'\Vert_{\mathbf D}^2
\]
\end{proof}

\begin{lemma}\label{lemma3}
Let
\begin{align*}
	T_{\gamma}(\bbb) &= \operatorname{Prox}_{\gamma, g}^{\mathbf D}(\bbb-\gamma \mathbf D^{-1} \nabla h(\bbb)),\\
	S_{\gamma}(\bbb) &= \operatorname{Prox}_{\gamma, g}^{\mathbf D}(\bbb-\gamma \mathbf D^{-1}\mathbf G_\bbb(\boldsymbol\xi)),\\
	\boldsymbol\eta &= \mathbf D^{-1}\mathbf G_\bbb(\boldsymbol\xi) - \mathbf D^{-1}\nabla h(\bbb).
\end{align*}
Then for $\bbb\in \mathcal{B}_1,$ and $\gamma>0,$
\begin{equation}\label{eq:a4}
 	\left\|T_{\gamma}(\bbb)-S_{\gamma}(\bbb)\right\|_{\mathbf D} \leq \gamma\|\boldsymbol\eta\|_{\mathbf D}
 \end{equation}
\end{lemma}

\begin{proof}[Proof of Lemma~\ref{lemma3}]
\begin{align*}
	\left\|T_{\gamma}(\bbb)-S_{\gamma}(\bbb)\right\|_{\mathbf D} &= \Vert\operatorname{Prox}_{\gamma, g}^{\mathbf D}(\bbb-\gamma \mathbf D^{-1} \nabla h(\bbb)) - \operatorname{Prox}_{\gamma, g}^{\mathbf D}(\bbb-\gamma \mathbf D^{-1}\mathbf G)\Vert_{\mathbf D}\\
	&\leq \Vert\gamma\mathbf D^{-1}\mathbf G - \gamma\mathbf D^{-1}\nabla h(\bbb)\Vert_{\mathbf D}\\
	&\leq \gamma\Vert\boldsymbol \eta\Vert_{\mathbf D},
\end{align*}
where the first inequality follows from Lemma~\ref{lemma:g_bound}-(ii).

\end{proof}

\begin{lemma}\label{lemma:supremum_finite}
Assume H4 and H8. Then $\sup_t \mathbb{E}[W^p(\boldsymbol\xi_t)]<\infty.$
\end{lemma}
\begin{proof}[Proof of Lemma~\ref{lemma:supremum_finite}]
As the conditional distribution of $\boldsymbol\xi_t$ given $\mathcal{F}_{t-1}$ is $P_{\bbb^{(t-1)}}(\boldsymbol\xi_{t-1},\cdot),$ so
\[
	\mathbb{E}[W^p(\boldsymbol\xi_t)] = \mathbb{E}[\mathbb{E}[W^p(\boldsymbol\xi_t)|\mathcal{F}_{t-1}]] = \mathbb{E}[P_{\bbb^{(t-1)}}W^p(\boldsymbol\xi_{t-1})]\leq \lambda\mathbb{E}[W^p(\boldsymbol\xi_{t-1})]+b.
\]
And by induction the proof is concluded.
\end{proof}
\begin{lemma}\label{lemma:approx_bound}
	Assume H1, H4, H7-(ii) and H8. There exist a constant $C$ such that w.p.1, for all $t\geq 0,$ $\Vert\boldsymbol\eta_t\Vert\leq CW(\boldsymbol\xi^{(t)}).$
\end{lemma}
\begin{proof}[Proof of Lemma~\ref{lemma:approx_bound}]
By definition,
\[
	\Vert\boldsymbol\eta_t\Vert = \Vert (\mathbf D^{(t)})^{-1}\mathbf G_{\bbb^{(t-1)}}(\boldsymbol\xi^{(t)}) - (\mathbf D^{(t)})^{-1}\nabla h(\bbb^{(t-1)})\Vert \leq \frac{1}{c_1}(\sup_{\bbb\in\mathcal{B}_1}\vert\mathbf G_\bbb\vert_W)W(\boldsymbol\xi^{(t)}) + \frac{1}{c_1} \sup_{\bbb\in\mathcal{B}_1}\Vert\nabla h(\bbb)\Vert.
\]
And the result follows as $\nabla h$ is Lipschitz and $W\geq 1.$
\end{proof}


\begin{lemma}\label{lemma:approx_bound3}
Assume H1, H4, H5 and H8. If $a_t\geq 0,$ for $t\geq 1,$ there exist a constant C such that
	\begin{equation}
	 	\left\Vert\sum_{t=1}^n a_t\Vert\boldsymbol\eta_t\Vert_{\mathbf D^{(t)}}^2\right\Vert_{L_2}\leq C\sum_{t=1}^n a_t
	\end{equation}
\end{lemma}
\begin{proof}[Proof of Lemma~\ref{lemma:approx_bound3}]
By Minkowski inequality,
\[
	\left\Vert\sum_{t=1}^n a_t\Vert\boldsymbol\eta_t\Vert_{\mathbf D^{(t)}}^2\right\Vert_{L_2}\leq C \sup_t\Vert\boldsymbol\eta_t\Vert_{L_4}^2\sum_{t=1}^n a_t\leq C\sum_{t=1}^n a_t,
\]
as the supremum is finite based on Lemma~\ref{lemma:supremum_finite} and Lemma~\ref{lemma:approx_bound}.
\end{proof}

\begin{lemma}\label{lemma:operator_T_bound}
Assume H1, H4, and H6. Then
\begin{equation}\label{eq:bound_T1}
	\sup _{\gamma \in(0,c_1 / L]} \sup _{\bbb \in \mathcal{B}_1}\left\|T_{\gamma}(\bbb)\right\|<\infty.
\end{equation}
If additional H7-(ii) holds, then there exist a constant $C$ such that for any $\bbb,\bbb^\prime\in\mathcal{B}_1,$ $\gamma,\gamma^\prime\in (0,c_1/L],$
\begin{equation}\label{eq:bound_T2}
	\left\|T_{\gamma}(\bbb)-T_{\gamma^\prime}(\bbb^\prime)\right\| \leq C(\gamma+\gamma\prime+\|\bbb-\bbb^\prime\|)
\end{equation}
\end{lemma}
\begin{proof}[Proof of Lemma~\ref{lemma:operator_T_bound}]
	As $\bbb_\star = T_\gamma(\bbb_\star)$ for any $\gamma>0.$ And
	\begin{align*}
		\Vert T_\gamma(\bbb) - \bbb_\star\Vert &= \Vert T_\gamma(\bbb) - T_\gamma(\bbb_\star)\Vert\\
		& = \Vert\text{Prox}_{\gamma,g}^{\mathbf D}(\bbb - \gamma\mathbf D^{-1}\nabla h(\bbb)) - \text{Prox}_{\gamma,g}^{\mathbf D}(\bbb_\star - \gamma\mathbf D^{-1}\nabla h(\bbb_\star))\Vert\\
		&\leq \frac{1}{c_1}\Vert \bbb-\gamma \mathbf D^{-1}\nabla h(\bbb) - \bbb_\star + \gamma\mathbf D^{-1}\nabla h(\bbb_\star)\Vert_{\mathbf D}\\
		&\leq (1+\frac{c_2}{c_1})\left(\Vert \bbb\Vert+\Vert\bbb_\star\Vert\right) < \infty,
	\end{align*}
	where the first and second inequality comes from Lipschitz proporty of $\text{Prox}_{\gamma,g}^{\mathbf D}$ (see H7-(ii)) and $\nabla h$, respectively. So we have (\ref{eq:bound_T1}) holds.
	
To prove (\ref{eq:bound_T2}), decompose $T_{\gamma}(\bbb)-T_{\gamma^\prime}(\bbb^\prime) = T_{\gamma}(\bbb)-T_{\gamma^\prime}(\bbb) + T_{\gamma^\prime}(\bbb) - T_{\gamma^\prime}(\bbb^\prime).$
	\begin{align*}
		\Vert T_{\gamma^\prime}(\bbb)-T_{\gamma^\prime}(\bbb^\prime)\Vert &\leq \frac{1}{c_1}\Vert\bbb-\gamma^\prime\mathbf D^{-1}\nabla h(\bbb) - \bbb^\prime + \gamma^\prime\mathbf D^{-1}\nabla h(\bbb')\Vert_{\mathbf D}\\
		&\leq \frac{c_2}{c_1}\Vert \bbb-\bbb^\prime\Vert + \frac{2\sup_{\bbb\in\mathcal{B}_1}\Vert\nabla h(\bbb)\Vert}{c_1}\gamma^\prime\\
		&\leq C(\gamma^\prime + \Vert\bbb-\bbb^\prime\Vert).
	\end{align*}
Since H6 and $\mathcal{B}_1$ is compact, $\sup_{\bbb\in\mathcal{B}_1}\Vert\nabla h(\bbb)\Vert<\infty.$

\begin{align*}
	\Vert T_\gamma(\bbb)-T_{\gamma^\prime}(\bbb)\Vert &= \Vert\text{Prox}_{\gamma,g}^{\mathbf D}(\bbb-\gamma\mathbf D^{-1}\nabla h(\bbb)) - \text{Prox}_{\gamma^{\prime},g}(\bbb-\gamma^\prime\mathbf D^{-1}\nabla h(\bbb))\Vert\\
	&= \Vert\text{Prox}_{\gamma,g}^{\mathbf D}(\bbb-\gamma\mathbf D^{-1}\nabla h(\bbb)) - \text{Prox}_{\gamma,g}^{\mathbf D}(\bbb)\Vert \\
	&+ \Vert\text{Prox}_{\gamma^{\prime},g}^{\mathbf D}(\bbb-\gamma^\prime\mathbf D^{-1}\nabla h(\bbb))  - \text{Prox}_{\gamma^\prime,g}^{\mathbf D}(\bbb)\Vert\\
	&+ \Vert\text{Prox}_{\gamma,g}^{\mathbf D}(\bbb) - \text{Prox}_{\gamma^\prime,g}^{\mathbf D}(\bbb)\Vert\\
	&\leq \frac{1}{c_1}\bigg(\sup_{\bbb\in\mathcal{B}_1}\Vert\nabla h(\bbb)\Vert(\gamma+\gamma^\prime)  + \Vert\text{Prox}_{\gamma,g}^{\mathbf D}(\bbb)-\bbb\Vert_{\mathbf D} + \Vert\text{Prox}_{\gamma^\prime,g}^{\mathbf D}(\bbb)-\bbb\Vert_{\mathbf D}\bigg)\\
	&\leq \frac{1}{c_1}\left(\sup_{\bbb\in\mathcal{B}_1}\Vert\nabla h(\bbb)\Vert  + c_2\sup_{\gamma\in(0,c_1/L]}\sup_{\bbb\in\mathcal{B}_1}\Vert\text{Prox}_{\gamma,g}(\bbb)-\bbb\Vert\right)(\gamma+\gamma')\leq C(\gamma+\gamma^\prime).
\end{align*}
The above inequality follows from assumption H7-(ii).
\end{proof}





\begin{proof}[Proof of Theorem~\ref{thm:thm2}]
~

By assumption $\bbb_\star\in \mathcal{B}_1,$ and $\bbb_\star = \argmin_{\bbb\in\mathcal{B}_1}F(\bbb):=\min F.$ Apply (\ref{eq:a3}) with $\bbb \leftarrow \bbb^{(t)}-\gamma_{t+1} \left(\mathbf D^{(t+1)}\right)^{-1} \mathbf G_{\bbb^{(t)}}(\boldsymbol\xi^{(t+1)}),$ $\boldsymbol\xi \leftarrow \bbb^{(t)},$ $\bbb^{\prime} \leftarrow \bbb_{\star}, \gamma \leftarrow \gamma_{t+1},\mathbf D \leftarrow \mathbf D^{(t+1)},$ we have
\begin{equation}\label{eq:thm1_1}
	\Vert\bbb^{(t+1)} - \bbb_\star\Vert_{\mathbf D^{(t+1)}}^2 \leq \Vert\bbb^{(t)} - \bbb_\star\Vert_{\mathbf D^{(t+1)}}^2 - 2\gamma_{t+1}\left(F(\bbb^{(t+1)}) - F(\bbb_\star)\right) - 2\gamma_{t+1}\left\langle\bbb^{(t+1)} - \bbb_\star, \boldsymbol\eta_{t+1}\right\rangle_{\mathbf D^{(t+1)}}.
\end{equation}

By rearranging (\ref{eq:thm1_1}), we have
\begin{align*}
	F(\bbb^{(t+1)}) - F(\bbb_\star)&\leq \frac{1}{2\gamma_{t+1}}\left(\Vert\bbb^{(t)} - \bbb_\star\Vert_{\mathbf D^{(t+1)}}^2 - \Vert\bbb^{(t+1)} - \bbb_\star\Vert_{\mathbf D^{(t+1)}}^2\right) - \left\langle\bbb^{(t+1)} - \bbb_\star, \boldsymbol\eta_{t+1}\right\rangle_{\mathbf D^{(t+1)}}\\
	&\leq \frac{1}{2}\left(\frac{1}{\gamma_{t+1}}-\frac{1}{\gamma_t}\right)\Vert\bbb^{(t)} - \bbb_\star\Vert_{\mathbf D^{(t+1)}}^2 - \frac{1}{2\gamma_{t+1}}\Vert\bbb^{(t+1)} -\bbb_\star\Vert_{\mathbf D^{(t+1)}}^2\\
	& + \frac{1}{2\gamma_t}\Vert\bbb^{(t)} - \bbb_\star\Vert_{\mathbf D^{(t+1)}}^2 - \left\langle\bbb^{(t+1)} - \bbb_\star, \boldsymbol\eta_{t+1}\right\rangle_{\mathbf D^{(t+1)}}
\end{align*}

Sum from $t=0,\ldots,n-1,$ and decompose
\[
	\langle\bbb^{(t)} - \bbb_\star, \boldsymbol\eta_t\rangle_{\mathbf D^{(t)}} = \left\langle\bbb^{(t)}-T_{\gamma_{t}}\left(\bbb^{(t-1)}\right), \boldsymbol\eta_{t}\right\rangle_{\mathbf D^{(t)}}+\left\langle T_{\gamma_{t}}\left(\bbb^{(t-1)}\right)-\bbb_{\star}, \boldsymbol\eta_{t}\right\rangle_{\mathbf D^{(t)}}.
\]
By (\ref{eq:a4}), we have $\left|\left\langle\bbb^{(t)}-T_{\gamma_{t}}\left(\bbb^{(t-1)}\right), \boldsymbol\eta_{t}\right\rangle_{\mathbf D^{(t)}}\right| \leq \gamma_{t}\left\|\boldsymbol\eta_{t}\right\|^{2}_{\mathbf D^{(t)}}$, so
\begin{equation}\label{obj_bound}
	\begin{aligned}
	\sum_{t=1}^n\left(F(\bbb^{(t)})-\min F\right) &\leq \sum_{t=1}^n \frac{1}{2}\left(\frac{1}{\gamma_{t}} - \frac{1}{\gamma_{t-1}}\right)\Vert\bbb^{(t-1)} - \bbb_\star\Vert_{\mathbf D^{(t)}}^2 + \frac{1}{2\gamma_0}\Vert\bbb^{(0)} - \bbb_\star\Vert_{\mathbf D^{(1)}}^2\\
	+ &\sum_{t=1}^{n-1}\frac{1}{2\gamma_t}\Vert\bbb^{(t)} - \bbb_\star\Vert_{\mathbf D^{(t+1)} - \mathbf D^{(t)}}^2 -\sum_{t=1}^{n}\left\langle T_{\gamma_{t}}\left(\bbb^{(t-1)}\right)-\bbb_{\star}, \eta_{t}\right\rangle_{\mathbf D^{(t)}}+\sum_{t=1}^{n} \gamma_{t}\left\|\boldsymbol\eta_{t}\right\|^{2}_{\mathbf D^{(t)}}
\end{aligned}
\end{equation}

Under the assumptions H6, the function $F$ is convex so that
\begin{equation}
	F(\bar\bbb_n) \leq \frac{1}{n}\sum_{t=1}^n F(\bbb^{(t)}).
\end{equation}

Denote $\Vert\cdot\Vert_{L_2} = \left(\mathbb{E}\Vert\cdot\Vert^2\right)^{1/2}.$ By (\ref{obj_bound}) and Minkowski inequality, we have there exists a constant $C>0,$ such that
\begin{align*}
	\Vert F(\bar\bbb_n) - \min F\Vert_{L_2}&\leq \frac{C}{n}\bigg(\sum_{t=1}^n\left|\frac{1}{\gamma_{t}}-\frac{1}{\gamma_{t-1}}\right| + \frac{1}{\gamma_0} + \sum_{t=1}^{n-1}\frac{1}{\gamma_t}\left\Vert\mathbf D^{(t+1)} - \mathbf D^{(t)}\right\Vert_{L_2}\\
	& + \left\Vert\sum_{t=1}^n\left\langle T_{\gamma_t}(\bbb^{(t)}),\boldsymbol\eta_t\right\rangle_{\mathbf D^{(t)}}
	\right\Vert_{L_2} + \left\Vert\sum_{t=1}^n\left\langle\bbb_\star,\boldsymbol\eta_t\right\rangle_{\mathbf D^{(t)}}\right\Vert_{L_2} + \left\Vert\sum_{t=1}^n\gamma_t\Vert\boldsymbol\eta_t\Vert_{\mathbf D^{(t)}}^2\right\Vert_{L_2}\bigg).
\end{align*}

By assumption, we assume $\gamma_t = C t^{-\alpha},\alpha\in (1/2,1],$
\[
	\sum_{t=1}^n\left|\frac{1}{\gamma_{t}}-\frac{1}{\gamma_{t-1}}\right|= O(n^{\alpha}),\quad \sum_{t=1}^{n-1}\frac{1}{\gamma_t}\left\Vert\mathbf D^{(t+1)} - \mathbf D^{(t)}\right\Vert_{L_2}= O(n^{\alpha})
\]
Apply $a_t = \gamma_t,$ in Lemma~\ref{lemma:approx_bound3}
\begin{equation}
 	\left\Vert\sum_{t=1}^n \gamma_t\Vert\boldsymbol\eta_t\Vert_{\mathbf D^{(t-1)}}^2\right\Vert_{L_2}\leq C\sum_{t=1}^n \gamma_t,
\end{equation}
and $\sum_{t=1}^n \gamma_t = O(n^{1-\alpha})$ for $\alpha\in (1/2,1),$ and $\sum_{t=1}^n \gamma_t = O(\ln n)$ for $\alpha=1.$

When $\boldsymbol\xi^{(t)}$ are sampled exactly, i.e., unbiased case, combine Lemma~\ref{lemma:operator_T_bound} and Proposition 18 of \cite{atchade2017perturbed}, there exists a constant $C$ such that
\[
	\left\Vert \sum_{t=0}^n\left\langle\mathbf A_{\gamma_{t+1}}(\bbb^{(t)}),\boldsymbol\eta_t\right\rangle_{\mathbf D^{(t)}}\right\Vert_{L_2}\leq C \sqrt{n}.
\]

Similarly, for the case of biased approximation, combine Lemma~\ref{lemma:operator_T_bound}, and Proposition 19 of \cite{atchade2017perturbed}, there exists a constant $C$ such that
\[
	\left\Vert \sum_{t=0}^n\left\langle\mathbf A_{\gamma_{t+1}}(\bbb^{(t)})\boldsymbol\eta_t\right\rangle_{\mathbf D^{(t)}}\right\Vert_{L_2}\leq C\left(1+ \sqrt{n} + \sum_{t=0}^n\gamma_t\right).
\]

In both cases, let $\mathbf A_{\gamma_t}(\bbb^{(t-1)}) = T_{\gamma_t}(\bbb^{(t-1)})$ and $\mathbf A_{\gamma_t}(\bbb^{(t-1)}) = I,$ we have

\[
	\left\Vert\sum_{t=1}^n\left\langle T_{\gamma_t}(\bbb^{(t-1)}),\boldsymbol\eta_t\right\rangle_{\mathbf D^{(t-1)}}
	\right\Vert_{L_2} =O(\sqrt{n})\text{ and} \quad \left\Vert\sum_{t=1}^n\left\langle\bbb_\star,\boldsymbol\eta_t\right\rangle_{\mathbf D^{(t-1)}}\right\Vert_{L_2} = O(\sqrt{n})
\]

Combine the above results and as $h$ is strongly convex, so there exist a $\mu>0,$ such that $F(\bar\bbb_n) - F(\bbb_\star)\geq \frac{\mu}{2}\Vert \bar\bbb_n - \bbb_\star\Vert^2,$ so we have
\[
	\mathbb{E}\Vert\bar\bbb_n - \bbb_\star\Vert^2\leq\left(\mathbb{E}\Vert\bar\bbb_n - \bbb_\star\Vert^4\right)^{1/2}\leq C\Vert F(\bar\bbb_n) - \min F\Vert_{L_2} \leq C n^{\alpha-1}.
\]
As $\alpha\in (1/2,1],$ by choosing $\alpha=1/2+\epsilon,\epsilon>0,$ we have the lowest bound $C n^{-\frac{1}{2}+\epsilon}.$
\end{proof}

\section{Additional Simulation Results}
\change{We provide an additional simulation study to (1) assess the estimation of the asymptotic variances of parameter estimates and (2) assess the point estimation of the covariance between latent variables.
We consider a similar confirmatory IFA setting as in the simulation study I, with two factors, twenty items (i.e., $K=2,$ $J=20$), and the same design matrix $\mathbf Q$. The intercept parameters and non-zero loading parameters are drawn i.i.d. from the standard normal and a uniform distribution over the interval $(0.5, 1.5),$ respectively. The variances of two factors are set to be 1 and the covariance is set to be 0.4.
For each of the three sample sizes $N=1000, 2000, 4000$, 50 independent datasets are generated. We then apply the proposed USP method with 1000 burn-in size, and 4000 total iterations. Note that we use a larger burn-in size and a larger number of iterations here to ensure accurate computation of the asymptotic variances, because they tend to be more difficult to compute than the point estimates. The results from the UPS algorithm are compared with those from a
 standard EM algorithm that uses 31 quadrature points for each dimension.

We approximate the observed Fisher information matrix using the approach given in Remark~\ref{rmk:biprod}. Based on the approximated Fisher information matrix, we obtain the standard errors of parameter estimates. The obtained standard errors are compared with those given by the EM algorithm. The results are given in Figure~\ref{fig:se_examine}. Each panel of Figure~\ref{fig:se_examine} corresponds to a combination of a sample size and a type of parameters (loadings/intercepts/covariance). For each dataset and each parameter, we obtain the standard errors of the parameter estimate from the UPS and EM algorithms, respectively. These standard errors are shown as a point in the scatter plot, where the x-axis gives the standard error from the EM algorithm and the y-axis gives the standard error from the USP method. As we can see, all the points concentrate along the diagonal line, suggesting that the standard errors from the two algorithms are very close to each other.

We further assess the estimation of the covariance between the latent variables. The results are given in Figure~\ref{fig:mse_sigma21}. For each sample size, we compute the squared difference between the estimate given by the USP algorithm and the true value ($\sigma_{12} = 0.4$) and visualize the squared errors from the 50 datasets using a box plot. We see that all the squared errors are quite small and they decrease when the sample size increases.}
%


\begin{figure}[H]
	\centering
	\includegraphics[width=0.8\linewidth]{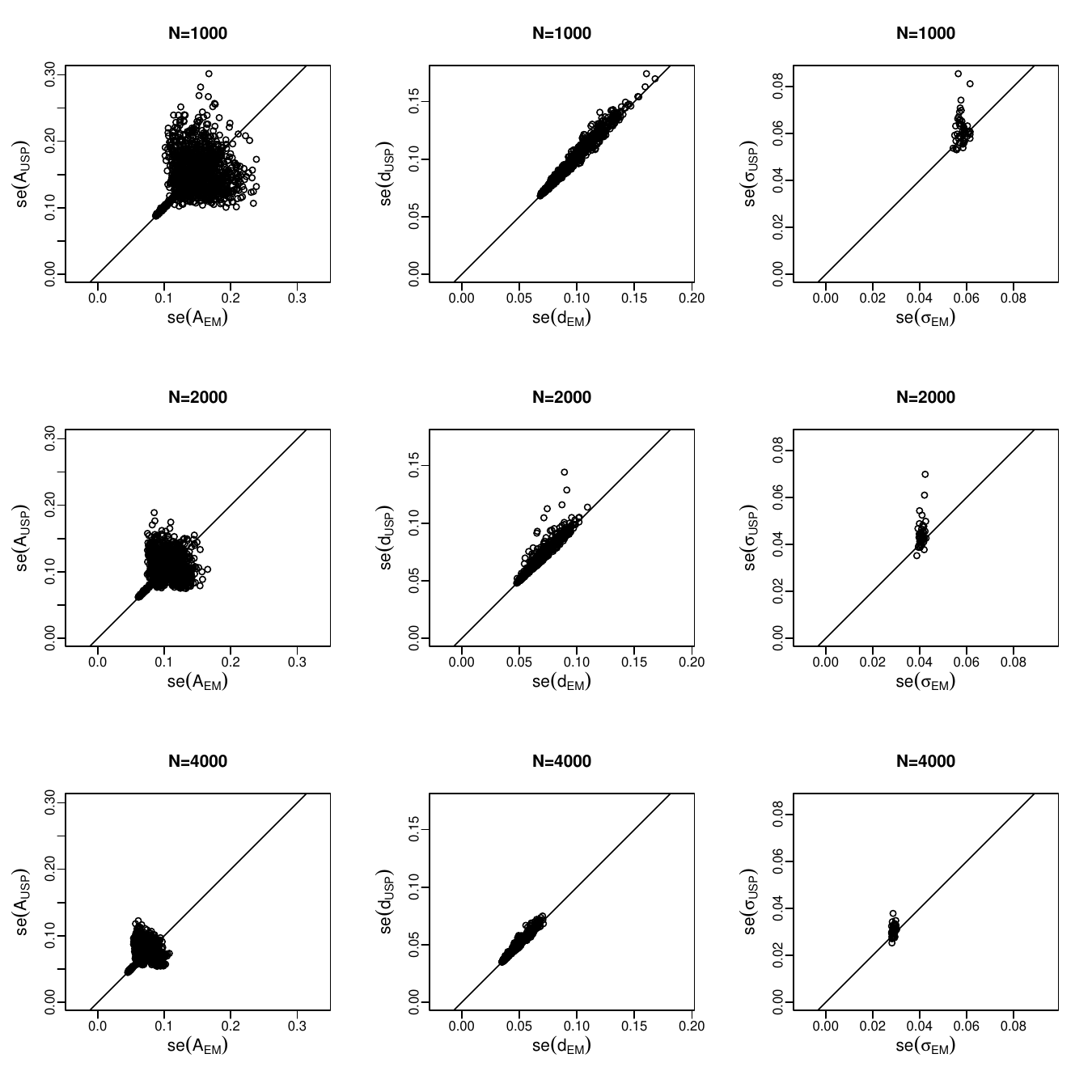}
	\caption{Scatter plots of standard error estimates for loading parameters $\mathbf A$, intercept parameters $\boldsymbol d$, and correlation parameter $\sigma$, from the EM method and the USP method under different sample sizes. The x-axis and y-axis represent standard error estimates from the EM and the USP method respectively. Each row corresponds to one sample size and each column corresponds to one type of parameter.}
	\label{fig:se_examine}
\end{figure}

\begin{figure}[H]
	\centering
	\includegraphics[width=.5\linewidth]{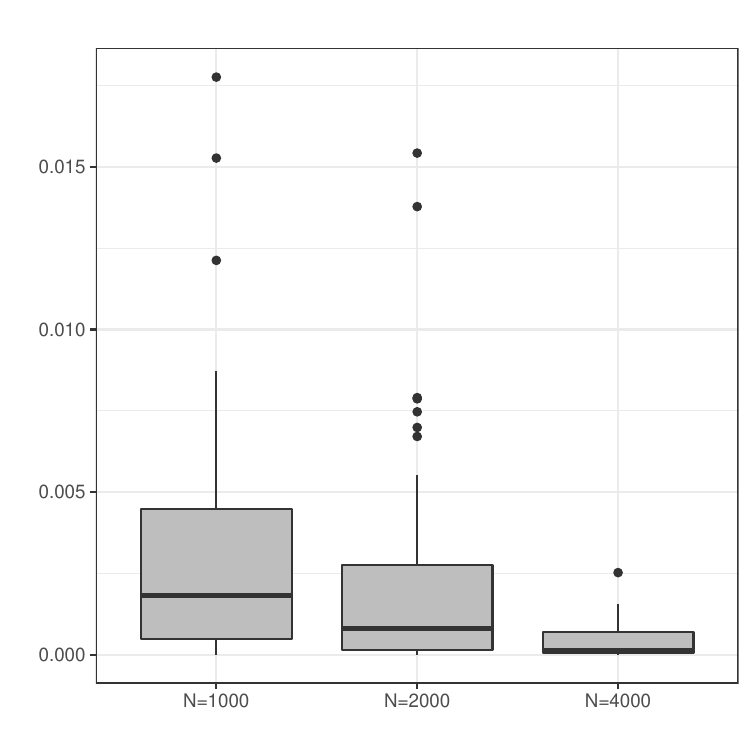}
	\caption{Box plots of squared errors for estimated correlation parameter $\sigma_{12}$ from the USP method.}
	\label{fig:mse_sigma21}
\end{figure}

\bibliographystyle{apacite}
\bibliography{ref}

\end{document}